\documentclass{article}
\usepackage{graphicx}

\newif\ifdraft
\draftfalse
\newif\ifaftersub
\aftersubfalse

\usepackage{setspace}

\usepackage[toc]{appendix}
\usepackage[margin=1in]{geometry}
\usepackage{subcaption}
\usepackage{enumerate,amsthm,amsmath,latexsym}
\usepackage[mathscr]{euscript}
\usepackage{bbm}
\usepackage{bbold,dsfont}
\usepackage{listings}
\usepackage{xcolor}
\usepackage{caption}
\usepackage{subcaption}
\usepackage{authblk}

\usepackage{hyperref}
\hypersetup{
    colorlinks,
    linkcolor={black!50!black},
    citecolor={green!50!black},
    urlcolor={blue!80!black}
}
\usepackage[nameinlink,capitalise]{cleveref}
\usepackage{framed}
\definecolor{shadecolor}{HTML}{F5F5F5}

\usepackage[utf8]{inputenc}
\usepackage{amsfonts}
\usepackage{tikz}
\usepackage{tkz-euclide}
\usepackage{xspace}
\usetikzlibrary{intersections}
\usetikzlibrary{positioning}
\usetikzlibrary{arrows}
\usetikzlibrary{shapes}
\usetikzlibrary{calc}

\usepackage{todonotes}

\newtheorem{theorem}{Theorem}[section]  

\newtheorem{lemma}[theorem]{Lemma}

\newtheorem{corollary}[theorem]{Corollary}
\newtheorem{assumptions}[theorem]{Assumptions}

\newtheorem{proposition}[theorem]{Proposition}

\usepackage{mathtools}
\definecolor{cricolor}{HTML}{F77888}

\usepackage[acronym]{glossaries}
\newacronym{rfid}{RFID}{Radio Frequency Identification}
\newacronym{irw}{IRW}{Idealized Random Walk}

\usepackage[normalem]{ulem}

\newcommand{\E}[1]{\mathbb{E}\left[#1 \right]}

\renewcommand{\Pr}{\mathbb{P}}

\newcommand{\omitthis}[1]{}

\usepackage{cancel}

\usepackage{stackengine}
\stackMath

\renewcommand{\epsilon}{\varepsilon}

\newcommand{\agent}{PA\xspace}

\title{Modeling Feasible Locomotion of Nanobots for Cancer Detection and Treatment}

\author[1]{Noble Harasha\thanks{nharasha@mit.edu}}
\author[2]{Cristina Gava\thanks{cristina.gava@kcl.ac.uk}}
\author[1]{Nancy Lynch\thanks{lynch@csail.mit.edu}}
\author[3]{Claudia Contini\thanks{c.contini@imperial.ac.uk}}
\author[2]{Frederik Mallmann-Trenn\thanks{frederik.mallmann-trenn@kcl.ac.uk}}

\affil[1]{Massachusetts Institute of Technology (MIT), Cambridge, MA, USA}
\affil[2]{King's College London, London, UK}
\affil[3]{Imperial College London, London, UK}

\begin{document}
\maketitle

\begin{abstract}
Deploying motile nanosized particles, also known as ``nanobots'', in the human body promises to improve selectivity in drug delivery and reduce side effects.
We consider a swarm of nanobots locating a single cancerous region and treating it by releasing an onboard payload of drugs at the site.
At nanoscale, the computation, communication, sensing, and locomotion capabilities of individual agents are extremely limited, noisy, and/or nonexistent.

We present a general model to formally describe the individual and collective behavior of agents in a colloidal environment, such as the bloodstream, for cancer detection and treatment by nanobots.
This includes a feasible and precise model of agent locomotion, inspired by actual nanoparticles that, in the presence of an external chemical gradient, move towards areas of higher concentration by means of self-propulsion.
We present two variants of our general model:  The first variant assumes an endogenous chemical gradient that is fixed over time and centered at the targeted cancer site; the second is a more speculative and dynamic variant in which agents themselves create and amplify a chemical gradient centered at the cancer site.
In both settings, agents can sense the gradient and ascend it noisily, locating the cancer site more quickly than via simple Brownian motion.

For the first variant of the model, we present simulation results to show the behavior of agents under our locomotion model, as well as {analytical results} to bound the time it takes for the  agents to reach the cancer site. We show that the agent's locomotion follows three distinct phases, determined by its distance from the cancer site.
For the second variant, simulation results highlight the collective benefit in having agents issue their own chemical signal.
The second variant of the model, while arguably more speculative in its agent capability assumptions, shows a significant improvement in runtime performance over the first variant, resulting from its chemical signal amplification mechanism.
\end{abstract}


\paragraph{Significance}
We present a mathematical model of nanorobots moving in a colloidal environment within the human body for the purpose of locating a single, targeted cancer site and delivering some localized treatment.
The capabilities and behavior of individual agents are inspired by actual chemotactic nanoparticles, making the model reasonably feasible.
We consider two distinct scenarios and accompanying variants of the general model: one in which there already exists a chemical signal for agents to follow, and the other which removes this assumption but has agents able to carry and release chemical payloads in order to form a chemical signal themselves.
While the latter setting is speculative, our results show a significant improvement in performance, possibly motivating the design of new nanoparticles.

\section{Introduction}

Motile nanoparticles suspended in a solution, or what are generally referred to as ``nanobots'', possess great potential in many medical applications as their scale allows for: penetrating biological barriers to otherwise unreachable regions of the body (e.g., the blood-brain barrier), unique maneuverability \textit{within} certain regions of the body, the deployment of swarms containing significantly high numbers of individual agents, and general precision.
We consider the dual problem of cancer detection and treatment in the context of nanomedicine, whose use holds significant promise.
 Nanobots are being studied and engineered to move within the human body and target sites of interest, such as a group of cancerous cells, in order to offer treatment by releasing the appropriate drugs once the target site is located and reached \cite{brigger2012nanoparticles, zhang2023advanced, ghosh2017stimuli}.
 
This offers a drug delivery solution which is more precise and selective, and thus less toxic to extraneous regions, especially in comparison to existing methods of treatment such as chemotherapy.
This promising approach, however, faces a number of challenges including limited, if not absent, control on the behavior of these robots as a consequence of their very small size.
If externally controlled movement is to be implemented, it often comes at the cost of having larger scale (e.g., microscopic) robots.

We therefore investigate the interesting and crucial challenge of nanobots finding and treating a cancerous site, \textit{distributively} and \emph{autonomously}.
In this work, we consider the process to be stochastic in nature, with nanobots performing a biased random walk, influenced by the chemical makeup and physics of the environment through which they navigate.

We first consider the situation---hereinafter referred to as the \textit{``passive agent model''}---in which there already exists a naturally occurring chemical signal coming from the targeted cancer site which agents can sense and follow.
Next, we consider a speculative scenario---hereinafter referred to as the \textit{``active agent model''}---in which there is no endogenous chemical signal to follow; agents must find the cancer site initially unaided, and then eventually create a chemical signal themselves for the later agents to use.
We formally describe the general models we propose for both settings in Section \ref{model}, with a particular focus on a model of agent locomotion that is based upon experimental data on actual motile nanoparticles.
We then present simulation and analytical results for the passive agent model in Section \ref{results1} and  simulation results for the active agent model in Section \ref{results2}.
The results for the active agent model show a significant runtime improvement as the total number of agents in the swarm increases.
We conclude the work with a brief discussion and suggestions for future work.
The full simulation code is available on \href{https://github.com/nobleharasha/nanobots-singleSite-ctsSim}{GitHub}.

\paragraph{Related work:}

At the end of the 20$^{th}$ century, the first ideas on nanotechnology and its potential applications in medicine emerged and laid the groundwork for the development of the field of nanomedicine. The works in \cite{crandall1996nanotechnology,freitas1999nanomedicine} present some of the first introductions to the concept of nanotechnology and the potential gains of its use for medical purposes, including drug delivery, diagnostics, and tissue engineering.
Here, we study the application of nanoparticles to drug delivery for cancer treatment, as explored in \cite{brigger2012nanoparticles, kostarelos2010nanorobots, zhang2023advanced}. 
The main motivation behind the use of nanoparticles here is the reduction of systemic side effects in favor of a treatment with the same level of efficacy, if not higher.
While immune cells, for example, are capable of long-distance sensing, most nanoparticles can be assumed to only have access to a very local subset of information about their environment.
The general problem of having a collective of agents locate an unknown target site when long-distance sensing is impossible is well-studied in the field of swarm robotics, including works which model social insects and other animal groups; at nanoscale, though, unique challenges arise.

Beyond limited sensing capabilities, precise control over the locomotion of individual nanobots is difficult, if not impossible.
Authors in \cite{gomez2021markov} consider nanobots to have zero locomotion capabilities with their position changing by passive displacement, simply following the flow of the circulatory system.
On the other hand are a series of works \cite{gwisai2022magnetic, martel2009flagellated}, operating from a centralized perspective, which achieve precise control over the locomotion of agents by employing external rotating magnetic fields.
Our study, however, assumes nanobots are able to move \textit{distributively} and \emph{autonomously} in a space, by performing some sort of random walk, one that is not entirely determined by the external forces at play in the given environment (e.g., blood flow).

The L\'evy flight random walk \cite{clementi2021search, fujisawalevy, reynoldslevy} is a popular model for the movement of simple biological agents, as it is mathematically elegant, has been observed among many actual animal species, and has been proven to produce nearly optimal results in the problem of foraging.
The L\'evy flight random walk can be very loosely characterized by long strides interspersed with random self re-orientation. When considering particles moving in a medium with a weak chemical gradient, the L\'evy flight random walk is a framework not fully dissimilar from \textit{chemotaxis}, described as the movement of organisms in response to a chemical stimulus, particularly, a chemical gradient.
Studies involving chemotaxis typically consider a colloidal solution---a mixture where one substance, made of nanoscopic or microscopic insoluble particles, is dispersed in another substance, which is usually termed the medium of the solution. When considering a clearly sensed, uniform gradient such as our case, we will show that the random walk our nanoparticles perform is a biased random walk.
The studies in \cite{golestanian2005propulsion, howse2007self, claudiavesicles, sanchez_catalase, sanchez_lm} are deeply relevant for chemotaxis by nanoparticles. 
\cite{golestanian2005propulsion, howse2007self, sanchez_catalase} look at the self-propulsion of nanoparticles in which the medium in which they are suspended does not present any particular gradient to follow; particles are able to move by the mechanism of chemotaxis, but this does \textit{not} result in a consistent or specific direction of movement by the collective.
Authors in~\cite{howse2007self} present a detailed model of the walk being performed, with the precise characterization of a drift velocity and a rotational velocity.

In contrast, in~\cite{claudiavesicles,sanchez_lm}, the authors present novel nanoparticles which are suspended in a solution with a nonuniform chemical gradient; via chemotaxis, the nanoparticles self-propel and follow the gradient, moving in a \textit{directed} fashion.
In~\cite{claudiavesicles}, authors---one of whom is an author of this work---present an in-vitro empirical analysis of nanoparticles' movement following the concentration gradient of a chemical dissolved in the colloidal solution. The chemical dissolved in the solution is glucose. 
The authors show how it is precisely the presence of this gradient that favors a specific type of walk from these particles, in which they ultimately tend towards areas of higher chemical concentration, noisily ascending the gradient.
The nanoparticles are asymmetrical, biocompatible vesicles containing glucose oxidase and catalase. 
In the presence of an external glucose gradient, these encapsulated enzymes react with the glucose, expelling products outwards asymmetrically (by nature of the vesicles' asymmetrical form), inducing a slip velocity on the surface of the particle in the opposite direction of the gradient and propelling it forward in the direction of the gradient \cite{claudiavesicles}.

Based on the empirical results in~\cite{claudiavesicles}, we define a mathematical model, our \emph{passive agent model}, which rigorously characterizes the movement of these nanoparticles.
We extend this to an \emph{active agent model}, with additional agent capability and environment assumptions, some speculative.
This model is also relevant to cancer detection, 
and is consistent with the locomotion of agents in our passive agent model based on~\cite{claudiavesicles}.

Basic signal amplification ideas similar to those in our active agent model were described and studied in~\cite{doi:10.1073/pnas.0610298104, doi.org/10.1038/nmat3049}.

\section{Model}\label{model}

We present a continuous space, discrete time general model for the problem of cancer detection and treatment by nanobots in the human body, including two variants of the model with 
distinct environment and agent capabilities assumptions to be outlined below.
In particular, we include a precise model for agent locomotion.
Figure \ref{fig:paramtable} is provided as a lookup table for the parameter notation used throughout the work in defining this general formal model (and its two variants).

\begin{table}
    \centering
    \begin{tabular}{ |p{1.5cm}||p{9cm}|  }
     \hline
     Notation & Parameter Description\\
     \hline
     \hline
     $n$ & total number of nanobot agents \\
     \hline
     $\epsilon$ & cancer site detection distance \\
     \hline
     $x^*, x_i^{(t)}, x^0$ & locations $\in \mathbb{R}^2$ of cancer site, agent $i$ at time $t$, and initial site, respectively \\
     \hline
     $\mu_i^{(t)}$ & vector pointing directly towards cancer site $\left(x^* - x_i^{(t)}\right)$ \\
     \hline
     $\phi(t)$ & current distance to cancer site at time $t$ $||\mu_i^{(t)}||_2$\\
     \hline
     $\theta_i^{(t)}$ & orientation vector $\in \mathbb{R}^2$ of agent $i$ at time $t$ \\
     \hline
     $\beta$ & angle $\in [-\pi, \pi)$ formed between $\mu_i^{(t)}$ and $\theta_i^{(t)}$ \\
     \hline
     $y^{(t)}$ & number of drug payloads dropped up to time $t$ \\
     \hline
     $z^{(t)}$ & number of signal chemical payloads dropped up to time $t$ (for active agent model) \\
     \hline
     $\gamma(\cdot)$ & signal chemical concentration function \\
     \hline
     $P, t^*$ & signal chemical gradient magnitude and steepness, respectively (for passive agent model) \\
     \hline
     $P, t_j^*$ & signal chemical payload size and time $j^{\text{th}}$ signal chemical payload dropped, respectively (for active agent model) \\
     \hline
     $D$ & diffusion coefficient \\
     \hline
     $b$ & orientation-bias parameter \\
     \hline
     $\alpha$ & step length \\
     \hline
     $\phi_{\text{max}}$ & boundary distance \\
     \hline
     
    \end{tabular}
    \caption{Summary table of notation for relevant parameters.}
    \label{fig:paramtable}
\end{table}

\subsection{Passive Agent Model}\label{model1}

A set of $n$ identical agents---nanobots---move in a two-dimensional Euclidean space $\mathbb{R}^2$.
Time is discretized.
We consider one, single cancer site that is concentrated in a single point in space.
The cancer site naturally produces some unique surface cell marker, which agents can bind to via the appropriate antibody which we assume they possess.
This allows agents to detect the presence of a nearby cancer site once they are within $\epsilon >0$ units of distance, with perfect accuracy; we consider them to have ``reached the cancer site'' at this point.
Each agent carries a payload of a drug for cancer treatment.
Once an agent is this $\epsilon$-distance away from the cancer site, it immediately releases its payload of drug, delivering its marginal treatment, and then for all practical purposes effectively ceases to exist in the environment.

In this ``passive agent model'', we also importantly assume there exists a specific global chemical gradient (of some chemical) 
centered at the cancer site, which agents can sense.
We call this chemical the ``signal chemical'', which can be imagined to be some naturally occurring entity which increases in concentration at closer distances to tumors and cancerous cells. Note that this gradient is time-constant, i.e., for any location $x$ in the space, the chemical concentration at $x$ remains constant over time.
The role of the signal chemical will be formalized by the movement model in Subsection~\ref{sec:move}.
No direct interaction or communication occurs between agents. 
An agent's movement is a function of only its previous state and the signal chemical that it is currently sensing in the vicinity of its current location.

We investigate the agents' ability to distributively and autonomously locate the cancer site and release their payloads (i.e., the treatment).
Given the random and noisy behavior of individual agents, our goal will be for at least $75\%$\footnote{This proportion is arbitrarily chosen such that the total treatment delivered is effective, while still being sufficiently less than $100\%$. In this way, drug overdosing is avoided with high probability and runtime is reasonably finite.} of the agents to deliver their marginal treatment.
The following quantities are involved in the process.

     The location of the cancer site is indicated with $x^*$ and we assume without loss of generality that $x^* = \overrightarrow{0}$ (the two-dimensional null-vector).
     The location of an agent $i$ (arbitrary index $i \in \mathbb{N}, i \in [0, n-1]$) at time $t$ is $x_i^{(t)} \in \mathbb{R}^2$, and its orientation is the vector $\theta_i^{(t)} \in \mathbb{R}^2$.
     The location of the initial site from which all agents begin is indicated with $x^0 \in \mathbb{R}^2$, i.e., $x_i^{(0)} = x^0$ $\forall i$.
     The number of agents that have released their drug payload up to time $t$ is $y^{(t)} \in \mathbb{N}$, with $y^{(0)}=0$.
    
    The concentration of the ``signal chemical'' at location $x \in \mathbb{R}^2$ is given by
    $\gamma(x) = \frac{P}{4 \pi D t^*} \cdot \text{exp}\left(-\frac{(||x - x^*||_2)^2}{4 D t^*}\right)$
    where $P$, $D$, and $t^*$ are constants. Imagining the signal chemical gradient to be the result of instantaneous point source diffusion frozen at a specific point in time, $P$, $D$, and $t^*$ are, respectively, the point mass amount, diffusion coefficient, and time since point mass introduction (i.e., time diffusing).\footnote{
    The constant $t^*$ here is a parameter representing the steepness of the gradient. It can be viewed as the amount of time that passed between the point mass introduction and when time is frozen, at which point we stop the diffusion process from evolving further and the gradient from changing.}

To be concrete, all parameters values hereafter are in SI units.
To be consistent with \cite{claudiavesicles}, we can think of one timestep in our model to be equivalent to one second.
If we assumed that agents moved faster, we could consider larger total areas of operation while still achieving reasonable runtimes.

\subsection{Active Agent Model}\label{model2}

Next we consider a speculative and more dynamic variant of the general model as just introduced, which we call the ``active agent model''.
We begin by removing the assumption that there already naturally exists some ``signal chemical'' gradient in the environment which is centered at the cancer site.
However, agents can still perfectly detect the presence of the cancer site within an $\epsilon$-distance and subsequently bind there to release their payload.
Now, half ($\lceil \frac{n}{2} \rceil$) of the agents carry a payload of the cancer treating drug, while the other half ($\lfloor \frac{n}{2} \rfloor$) carry a payload of some artificial ``signal chemical'' which \textit{all} agents can sense.
The artificial signal chemical diffuses perfectly radially out from the cancer site via \textit{instantaneous point-source diffusion} \cite{diffusion} (see the new $\gamma(\cdot)$ below), forming a global signal chemical gradient centered at the cancer site.
This chemical gradient will serve the same purpose as the signal chemical gradient in the passive agent model, being followed by agents to locate the cancer site more efficiently.
Before, the signal chemical gradient was endogenous and constant over time; now, it is created by the agents themselves and changes over time ``passively'' as a result of diffusion and ``actively'' as a result of more agents releasing their payloads.

We redefine $y^{(t)}$ and introduce the quantity $z^{(t)}$ as follows:
The number of agents that have released their \textit{drug} payload up to time $t$ is $y^{(t)} \in \mathbb{N}$, with $y^{(0)}=0$.
The number of agents that have released their \textit{signal chemical} payload up to time $t$ is $z^{(t)} \in \mathbb{N}$, with $z^{(0)}=0$.
Imagine one agent releases its signal chemical payload of amount $P$ at the cancer site at time $t^*$. 
It will immediately begin diffusing, with the concentration of this individual payload, at location $x$ at time $t$, being a function of the distance from the cancer site $||x - x^*||_2$ and the time since its release $(t - t^*)$: $\frac{P}{4 \pi D (t - t^*)} \text{exp}\left(-\frac{(||x - x^*||_2)^2}{4 D (t - t^*)}\right)$ \cite{diffusion}, where $D$ is the diffusion coefficient. We simplify by assuming that the diffusion of each agent's payload is independent---so that the concentrations are additive.
Thus, the concentration of the signal chemical at time $t$ at location $x \in \mathbb{R}^2$ is therefore given by $\gamma^{(t)}(x) = \frac{P}{4 \pi D} \sum_{j=1} ^{z^{(t)}} \frac{1}{t - t_j^*} \text{exp}\left(-\frac{(||x - x^*||_2)^2}{4 D (t - t_j^*)}\right)$
where $t_j^*$ is the time at which the $j$'th agent \textit{with a signal chemical payload} to reach the cancer site released their payload.

Our goal now will be for $75\%$ of the agents \textit{with a cancer drug payload} to deliver their marginal treatment.
Besides the specific items noted above, the rest of the general model is left unchanged from the passive agent model in Subsection \ref{model1}.
As more agents reach the cancer site, they maintain, if not amplify, the signal chemical gradient, thus helping other agents find the cancer site more efficiently over time. 
We investigate the impact of this collective behavior.

\subsection{Orientation-Biased Movement Model}
\label{sec:move}

We now describe the update step for the locomotion of an individual agent, which is the same for the passive and active agent models.
Notably, in the active agent model, \textit{all} agents follow the same movement model regardless of which specific chemical payload they possess.
The only difference in the implementation of this update step between the two model variants is the specific signal chemical concentration function $\gamma(\cdot)$ being used, which is defined in Subsections \ref{model1} and \ref{model2}, respectively.

All of the steps traverse the same distance, but there is a bias of moving towards the cancer site.
As a result, we call this movement model the \textit{``Orientation-Biased Model''}.
Consider some agent $i$ and let $\mu_i^{(t)} \coloneqq x^* - x_i^{(t)}$ be the vector pointing from the current agent location to the target location.
Let $\phi(t)\coloneqq||\mu_i^{(t)}||_2$ be the current Euclidean distance between the agent and the cancer site.
Let $\theta_i^{(t)}$ be the orientation vector representing the direction agent $i$ moves in, relative to the cancer site. 
We define $\beta$ to be the angle formed by $\theta_i^{(t)}$ and $\mu_i^{(t)}$.
We model the noise in the movement direction by drawing $\beta \sim \mathcal{N}(0, \sigma^2)_{[-\pi,\pi)}$, i.e., truncated normal distribution, 
where $\beta \in [-\pi, \pi)$---in expectation, the agent moves toward the cancer cell. 
See Figure~\ref{fig:directions} for a depiction.

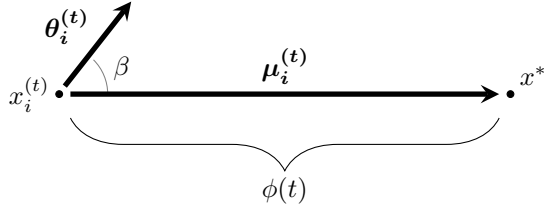
\begin{figure}
    \centering
    
    \begin{tikzpicture}[scale=1.5]
    \coordinate[label=above right:$x^*$] (O) at (0,0);
    \coordinate[label=left:$x_i^{(t)}$] (A) at (-4,0);
    \coordinate[] (B) at (-3, 1);
    \draw[line width=2pt,black,-stealth](-3.95,.071)--(-3.29-.071,.75+.071);
    \coordinate[label=left:$\boldsymbol{\theta_i^{(t)}}$] (_) at (-3.65,.6);
    \draw[line width=2pt,black,-stealth] (-3.9,0)--(-0.1,0);
    \coordinate[label=left:$\boldsymbol{\mu_i^{(t)}}$] (_) at (-1.7,.25);
    \draw[decorate,decoration={brace,amplitude=20pt,raise=0.5pt,mirror},yshift=0pt] (-3.9,-.2) -- (-0.1,-.2) node[midway,yshift=-28pt]{$\phi(t)$};
    \draw[fill=black] (A) circle (.03cm);
    \draw[fill=black] (O) circle (.03cm);
    
    \tkzMarkAngle[size=0.43cm, opacity=.5](O,A,B)
    \tkzLabelAngle[pos = 0.6](O,A,B){$\beta$}

    \end{tikzpicture}
    \caption{
    Visual representation of quantities defined for the orientation-biased movement model for an agent $i$. Precisely, the agent's location at time $t$ $x_i^{(t)}$, the cancer's location $x^*$, the orientation vector $\theta_i^{(t)}$, the agent-target location vector $\mu_i^{(t)}$, and the angle formed $\beta$. Recall that $\phi(t) = ||\mu_i^{(t)}||_2$.}
    \label{fig:directions}
\end{figure}
Recall that Table~\ref{fig:paramtable} lists all parameters.
We assume the underlying normal distribution has variance $\sigma^2 = \left ( b \cdot \biggl | \frac{\mathrm{d}}{\mathrm{d} \, \phi(t) \,} \left(\gamma^{(t)}x_i^{(t)}\right) \biggr | \right)^{-1}$ for some constant $b \in \mathbb{R}_{\geq 0}$, where we note that the variance scales inversely with the derivative of $\gamma(\cdot)$ with respect to the distance to the cancer site $\phi(t)$: $\frac{\mathrm{d}}{\mathrm{d}\phi}\gamma(\cdot)$.  
The steeper the signal chemical gradient, the more biased towards zero $\beta$ is.
Since the agent's movement direction is rotated by $\beta$, we can calculate the orientation vector $\theta_i^{(t)}$ as a function of $\beta$ as follows (representing $\mu_i^{(t)}$ as a column vector):
$\theta_i^{(t)} = 
\begin{bmatrix}
\cos{\beta} & -\sin{\beta} \\
\sin{\beta} & \cos{\beta}
\end{bmatrix} \mu_i^{(t)}.$
That is, there is always a bias to orient towards the cancer site, with a greater bias as the local chemical gradient is steeper.
This ``orientation-bias'' also increases as $b$ increases.
In fact, as $b$ approaches infinity, the movement becomes taking the shortest, straight-line path to $x^*$, and as $b$ approaches zero, the movement becomes the simple random walk\footnote{Hereafter, the `simple', `fully', or `unbiased' random walk in $\mathbb{R}^2$ all always refer to the movement model in which an agent's scalar orientation is sampled uniformly at random, i.e., $\beta\sim\mathcal{U}(-\pi,\pi)$.}.
Agent $i$'s position is then updated as follows, taking a step of length $\alpha$ in the direction of its orientation vector:
\begin{equation}\label{eq:theta} x^{(t)}_i = x^{(t-1)}_i + \alpha \cdot \frac{\theta_i^{(t)}}{||\theta_i^{(t)}||_2}.\end{equation}
For the rest of this paper, we will assume a bounded space, such that agents are always within $\phi_{\text{max}}$ units of distance from the cancer site $x^*$, i.e., space is a disk of radius $\phi_{\text{max}}$ centered at $x^*$. 
At each timestep, each agent attempts to follow the above update step until its new location is indeed within the given boundary.\footnote{For sufficiently small $\alpha$ a valid new location will eventually be achieved with probability $1$.  
We assume that all of these attempts begin from the same starting location, and happen within one timestep.}
Despite space being bounded, note that $\gamma(\cdot)$ models diffusion in an unbounded space; imposing this $\phi_{\text{max}}$ boundary is a simplification we choose to make simply in order to aid convergence/termination in simulations.

Figure~\ref{fig:ex} gives a simulated example of an agent moving under our model.
Recall that an agent's orientation-bias, as we call it, is a function of the derivative of the global signal chemical gradient $\gamma(\cdot)$, which is plotted in Figure~\ref{fig:exGrad} for this particular example simulation, with respect to the distance to the cancer site.
For $||x-x^*||_2 \in [0.004,0.01]$, the derivative is close to zero; for $||x-x^*||_2 \in [0.00075,0.004]$, the derivative is non-negligible in magnitude; and for $||x-x^*||_2 \in [0,0.00075]$, the derivative again flattens to zero.
This is reflected directly in the resulting agent's motion as plotted in Figure \ref{fig:exRun}: the motion begins as something close to a simple random walk or Brownian motion, becomes more biased/favorable once the agent reaches a distance of around $0.004$, and finally returns to being mostly fully random once extremely close to the cancer site.
This trend in the derivative and resulting behavior is a direct consequence of the form of the signal chemical concentration function $\gamma(\cdot)$ which is very roughly $\exp(-x^2)$, and will always be seen to varying extents regardless of parameter settings.

\begin{figure*}[t!]
    \centering
    \begin{subfigure}[t]{0.48\textwidth}
        \centering
        \includegraphics[width=.9\textwidth]{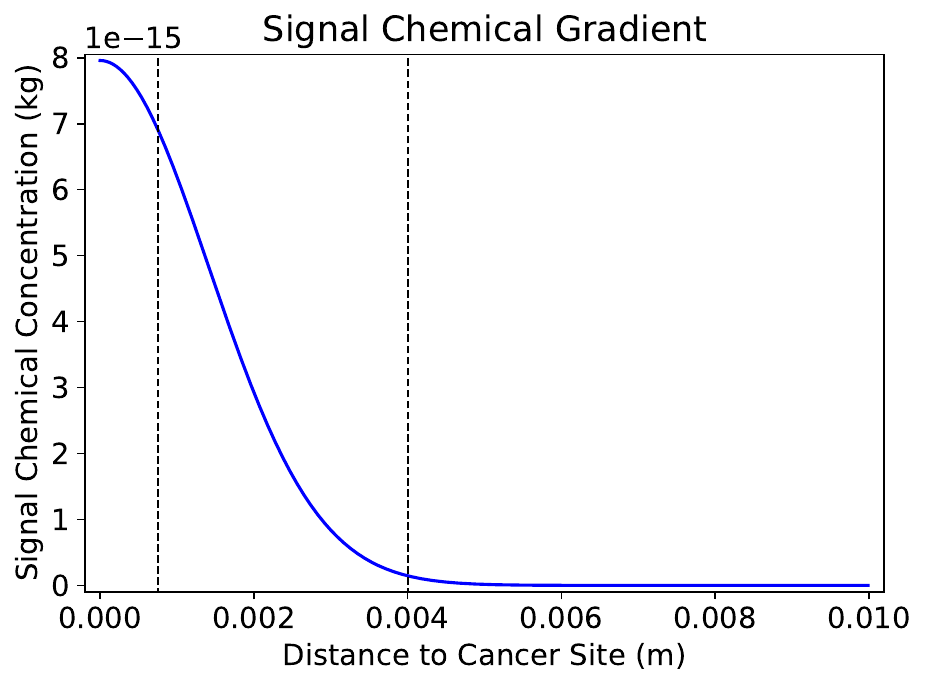}
        \caption{}
        \label{fig:exGrad}
    \end{subfigure}%
    ~ 
    \hfill
    \begin{subfigure}[t]{0.48\textwidth}
        \centering
        \includegraphics[width=.9\textwidth]{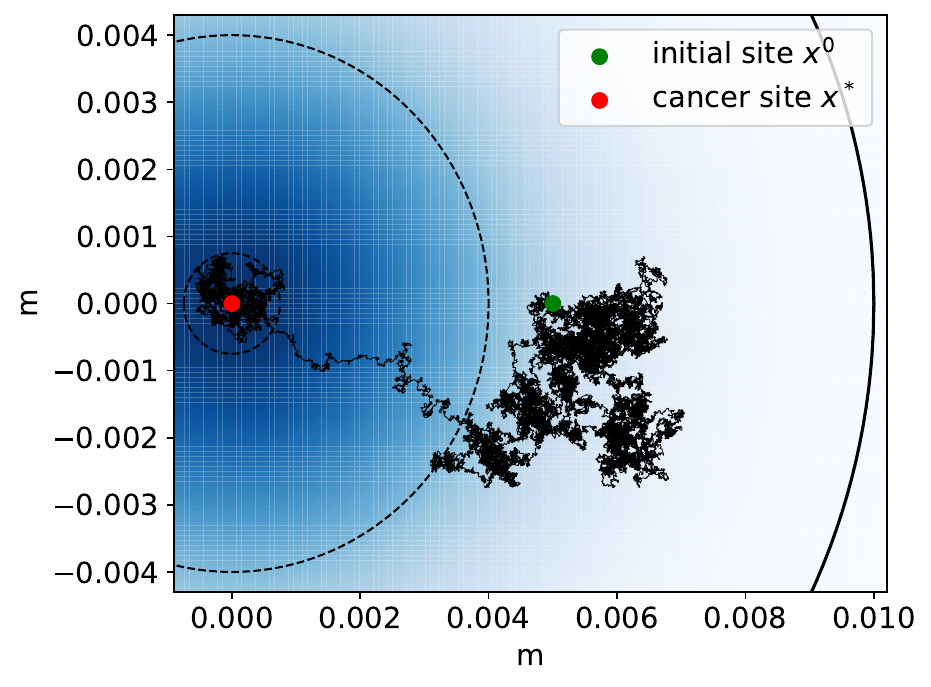}
        \caption{}
        \label{fig:exRun}
    \end{subfigure}
    \caption{\textit{(Passive agent model)} Example simulation run of one agent following our Orientation-Biased movement model with parameters $x^0=(0.005,0), \alpha = \epsilon = 2\cdot10^{-5}, P=10^{-19}, D=10^{-10}, t^*=10^{4}, b=10^{11}, \phi_{\text{max}}=0.01$ (SI units). (a) shows the concentration of signal chemical $\gamma(x)$ plotted as function of distance to cancer site $||x-x^*||_2$, given chemical gradient's radial symmetry. (b) shows the agent's position over $\approx 70000$ timesteps plotted as black trajectory. Signal chemical concentration shown in blue color map. Rightmost solid black curve is $\phi_{\text{max}}$ boundary. Dashed lines in both figures show the same reference distances of $0.00075$ m and $0.004$ m from the cancer site, respectively, and approximate the notable distances $d_2$ and $d_4$ defined in Subsection~\ref{sec:distances}.
    }
\label{fig:ex}
\end{figure*}

\subsection{Model Validation}

Our model, regarding the locomotion of individual agents or nanobots, is based upon the actual nanoscopic vesicles in \cite{claudiavesicles}.
Therefore, it is likely that some version of our model can be implemented, and can have real-world impact.
We now argue that our model resembles the locomotion of the actual nanoparticles. 
Figure \ref{fig:val} provides both a quantitative and qualitative comparison between the actual particles and the agents under our model, albeit in a single, specific environment setting.

The signal chemical gradient here is linear with respect to the distance to the cancer site, yielding a constant variance $\sigma^2$ at all points in time
and locations in space according to our Orientation-Biased Model of movement.
The histogram(s) on the left of Figure \ref{fig:oriVelo} show that the distribution of scalar orientations matches extremely closely between the actual particles in \cite{claudiavesicles} and our model.
The right panel of Figure \ref{fig:oriVelo} shows the average step sizes (effective velocities) of the agents. In contrast to the actual particles studied in \cite{claudiavesicles}, where more favorable orientations correlate with larger average step sizes between 2.8 and 5.5 (abstract) units, our model, for the sake of elegance and forthcoming mathematical analysis, assumes a fixed step size of 4.6 (abstract) units.
Figure \ref{fig:trajs} shows the corresponding agent trajectories in which our model exhibits agent movement which is reasonably similar to the actual particles, with both having agents noisily ascend the external chemical gradient.

\begin{figure*}
    \centering
    \begin{subfigure}[t]{0.48\textwidth}
        \centering
        \includegraphics[width=.9\textwidth]{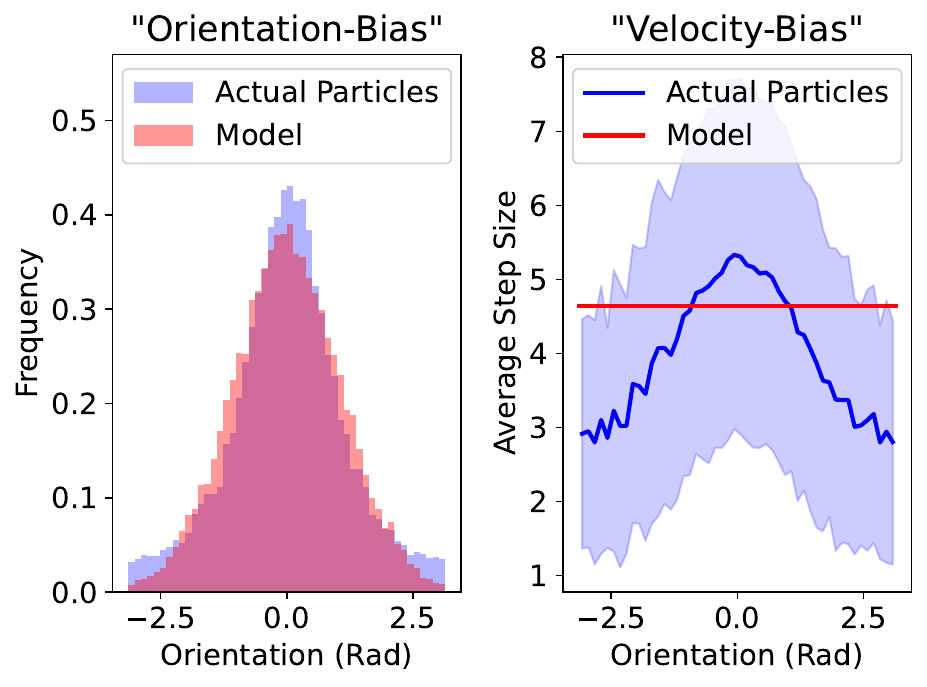}
        \caption{}
        \label{fig:oriVelo}
    \end{subfigure}%
    \hfill
    \begin{subfigure}[t]{0.48\textwidth}
        \centering
        \includegraphics[width=.9\textwidth]{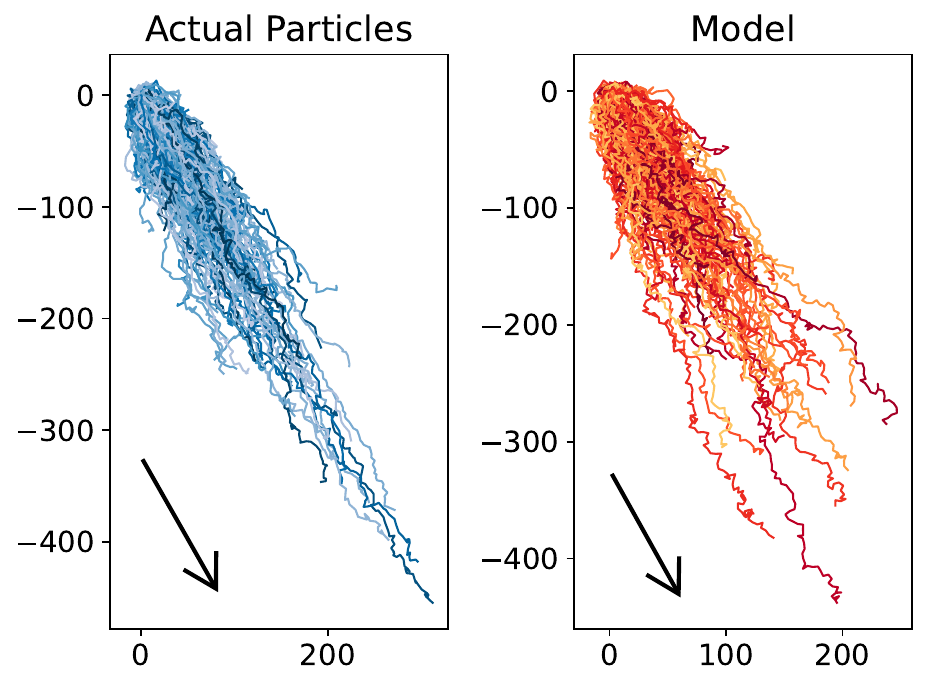}
        \caption{}
        \label{fig:trajs}
    \end{subfigure}
    \hfill
    \caption{Comparison between the actual nanoparticles of \cite{claudiavesicles} in vitro experiment and agents under our model in simulation. In both, agents are moving in response to an external chemical gradient which is linear with respect to distance and time-constant. 
    Note that units of distance are arbitrary/abstract here. (a)-Left: histograms of agent orientations over all timesteps. Orientation of zero points in same direction as chemical gradient increases. (a)-Right: average units of distance traveled by agents in one timestep (i.e., step size) over all timesteps. Shaded region is standard deviation. (b): single agent trajectories normalized to start at $(0,0)$. Chemical gradient increases in direction of black arrow.
    }
\label{fig:val}
\end{figure*}

\section{Results for the Passive Agent Model}\label{results1}

We present analytical results for the time required for a single agent to reach the cancer site in Subsection~\ref{sec: analytical}.
Still looking at the hitting time of a single agent, we present simulation results in 
Subsection~\ref{sec: simulation-single-agent}.
{We then consider a swarm with multiple agents and present both analytical and simulation results in Subsection~\ref{sec:multi}.
For all of our simulations, we fix $\phi_{\text{max}} = 0.05$ m, $\epsilon = \alpha = 2\cdot 10^{-5}$ m, and $D=10^{-10} \text{m}^2/\text{s}$.

\subsection{Analytical Results for a Single Passive Agent}
\label{sec: analytical}

We first consider the case of a single agent ($n=1$).
This is a simpler, related problem that may give some intuition for the behavior of the larger swarm/system, as agents' random walks in the passive agent model are independent.

Figure \ref{fig:expDP} shows the expected value of the progress made in a single step in the distance remaining to the cancer site, under different signal chemical gradients.
The exact expected value calculations were derived from \eqref{eq:triangle}.
Across all settings, the expected progress is close to zero from farther away, increases as the agent approaches the cancer site, and then decreases towards zero once more when extremely close.
This  trend mirrors our discussion of Figure~\ref{fig:ex}; we are in fact looking at the same phenomenon, just now using more precise analytical results to come to the same conclusion.
With larger values of $t^*$, i.e., as the signal chemical is less steep (or diffuses for longer), agents begin making nonnegligible, positive progress from farther away as the gradient now reaches further, but their absolute maximum expected progress decreases as the gradient is flattening.
Imagining the signal chemical gradient here to actually be dynamic and diffusing, the chemical signal expands but weakens over time.

\begin{figure}[htbp]
  \centering

  \begin{minipage}[b]{0.45\textwidth}
    \centering
    \begin{tikzpicture}[scale=1.5]
      \coordinate[label=above right:$x^*$] (O) at (0,0);
      \coordinate[label=left:$x_i^{(t)}$] (A) at (-4,0);
      \coordinate[label=above:$x_i^{(t+1)}$] (B) at (-3, 1);
      \draw (O)--(A)--(B)--cycle;
      \draw[line width=2pt,black,-stealth](-4.071,.071)--(-3.25-.071,.75+.071);
      \draw (-3.7, .5)--(-4.4,.5); 
      \coordinate[label=left:$\boldsymbol{\theta_i^{(t)}}$] (_) at (-4.3,.5);
      \draw[line width=2pt,black,-stealth](-4,-.1)--(0,-.1) node[below]{$\boldsymbol{\mu_i^{(t)}}$};
      \draw[decorate,decoration={brace,amplitude=20pt,raise=0.5pt},yshift=0pt] (A) -- (B) node[midway,yshift=19pt, xshift =-17pt]{$\alpha$};
      \draw[decorate,decoration={brace,amplitude=20pt,raise=0.5pt,mirror},yshift=0pt] (A) -- (O) node[midway,yshift=-28pt]{$\phi(t)$};
      \draw[fill=black] (A) circle (.03cm);
      \draw[fill=black] (B) circle (.03cm);
      \draw[fill=black] (O) circle (.03cm);
      \tkzLabelSegment[above=2pt](O,B){\textit{$\phi(t+1)$}}
      \tkzMarkAngle[size=0.43cm, opacity=.5](O,A,B)
      \tkzLabelAngle[pos = 0.6](O,A,B){$\beta$}
    \end{tikzpicture}
    \caption{
      \textit{(Passive agent model)} A single step, from time $t$ to time $t+1$, of an agent $i$ in the passive agent model with fixed scalar orientation $\beta$ and step length $\alpha$.
      Recall that $\phi(t)\in\mathbb{R}$ is the distance the agent is away from the cancer site at time $t$ ($\phi(t+1)$ is analogously defined, for time $t+1$). 
      Bold symbols represent vectors, non-bold symbols represent scalar quantities or points.
    }
    \label{fig:triangleOnestep}
  \end{minipage}
  \hfill
  \begin{minipage}[b]{0.45\textwidth}
    \centering
    \includegraphics[width=0.95\linewidth]{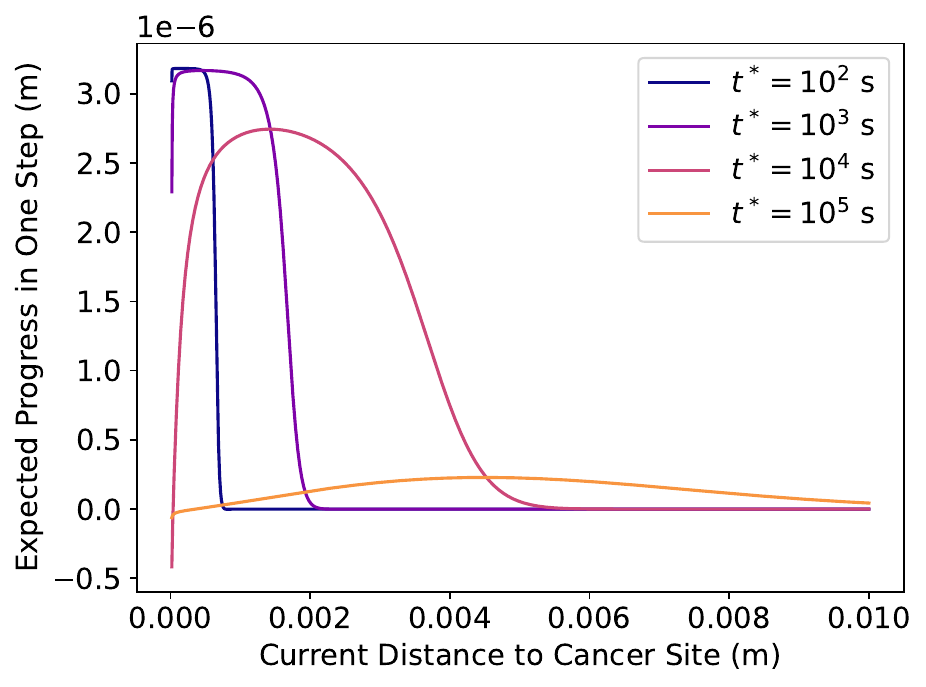}
    \caption{
      \textit{(Passive agent model)} Expected progress made in a single step towards the cancer site ($\E{\phi(t)-\phi(t+1)}$) by an agent under our passive agent model from varying current distances, and under different signal chemical gradients.
      Fixed parameters include $\alpha=\epsilon=2\cdot10^{-5}, P=10^{-19},D=10^{-10},b=10^{12},\phi_{\text{max}}=0.01$ (SI units).
    }
    \label{fig:expDP}
  \end{minipage}

\end{figure}

\paragraph{Expected Progress in a Single Step}\label{sec:singleStepProg}
We begin by deriving the expected progress made in a single step in \eqref{eq:triangle}, for use in the analysis to come.
For some agent $i$, with location $x_i^{(t)}$ at time $t$, we let $\phi\left(x_i^{(t)}\right)$
 be its current distance to the cancer site, i.e., 
 \[ \phi\left(x_i^{(t)}\right) = ||x_i^{(t)}-x^*||_2. \]
 When clear from the context which agent is meant, we abuse notation and simply write $\phi(t)$.
During timestep $t$, the agent moves with scalar orientation $\beta$ relative to the cancer site (see Figure~\ref{fig:triangleOnestep}).
Let $\Delta\phi(t+1) \coloneqq \phi(t) - \phi(t+1)$, i.e., $\Delta\phi(t+1) > 0$ implies the agent made positive progress towards the cancer site during timestep $t$ and is now closer.
The Law of Cosines gives $\phi(t+1)^2 = \phi(t)^2 + \alpha^2 - 2\phi(t)\alpha\cos(\beta)$, which yields $\Delta \phi(t+1) = \phi(t) - \sqrt{\phi(t)^2 - 2 \phi(t) \alpha\cos(\beta) + \alpha^2}$.
We write $\Delta  \phi_\beta$ to express the change conditioned on $\beta$.
Recall, we assume $\beta \sim \mathcal{N}(0, \sigma^2)_{[-\pi,\pi]}$, i.e., truncated normal distribution, 
where $\beta \in [-\pi, \pi)$ and \[\sigma^2 =\sigma^2_t= \frac{1}{ b \cdot \Bigl | \frac{d}{d \, \phi(t)} \gamma\left(x_i^{(t)}\right) \Bigr | } = \left( \frac{bP\phi(t)}{8\pi D^2 t^{*^2}} \exp\left( -\frac{\phi(t)^2}{4Dt^*} \right) \right)^{-1},\] as defined in the Orientation-Biased movement model.
Note that, abusing notation slightly, we have \( \lim_{\phi(t) \to 0^+} \sigma(\phi(t)) = \infty \) and \( \lim_{\phi(t) \to \infty} \sigma(\phi(t)) = \infty \), 
but \( \sigma(\phi(t)) \) attains a global minimum at a finite, positive value of \( \phi(t) \). As we will see in our analysis, the consequence of this is that the agent first (as a function of $\phi(t)$) behaves almost like a simple random walk, then like a biased random walk, and finally again like a simple random walk.
For ease, we hereafter let $r \coloneqq \frac{bP}{8\pi D^2 t^{*^2}} > 0$---i.e., the above $\sigma^2 = \left( r \phi(t) \exp\left(-\frac{\phi(t)^2}{4Dt^*} \right) \right)^{-1}$.
In \eqref{eq:triangle}, we  derive the expected progress made in a single step in the distance remaining to the cancer site.
%
%
Recall $\Delta\phi(t+1) =\phi(t) - \phi(t+1)$.
By symmetry, we can restrict $\beta \in [0, \pi]$. 
For some appropriately chosen constant $c$, we have
    \begin{align}
        \mathbb{E}[\Delta\phi(t+1)] &= 
        c \int_{0}^{\pi} \frac{\Delta\phi_\beta }{\sqrt{2\pi\sigma^2}} \exp\left(-\frac{\beta^2}{2\sigma^2} \right) d\beta \nonumber\\
        &= 
        \frac{c}{\sqrt{2\pi\sigma^2}} \int_{0}^\pi \left( \phi(t) - \sqrt{\phi(t)^2 - 2 \phi(t) \alpha\cos(\beta) + \alpha^2} \right) \exp\left(-\frac{\beta^2}{2\sigma^2}\right) d\beta \label{eq:triangle}
 \end{align}

Recall from Figures \ref{fig:exRun} and \ref{fig:expDP} the following general trend in agent motion in the passive agent model: fully random when far away from the cancer site, more biased when approaching, and finally fully random once again when extremely close.
We will break up the analysis into three phases, in accordance with the above progression, as illustrated in Figure \ref{fig:beauty}.

\definecolor{gr}{HTML}{008000}
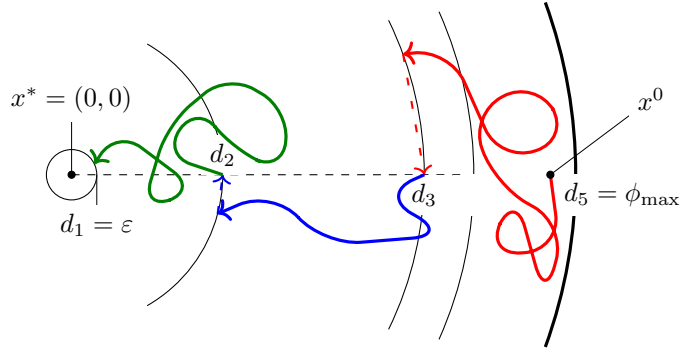
\begin{figure}
    \centering
    \begin{tikzpicture}[scale=.67]
        \filldraw[] (0,0) circle (2pt);
        \draw[very thin] (0,0)--(0,1.5) node[fill=white]{$x^*=(0,0)$};
        
        \draw[dashed] (0,0) -- (10,0);
        \draw[] (0,0) circle (.5);
        \draw[very thin] (.5,0)--(.5,-1) node[fill=white] {${d_1}=\epsilon$};
        \draw[] (1.5, -2.6) arc (-60:60:3) (3,.4) node[fill=white] {${d_2}$};
        \draw[] (6.29, -3.07) arc (-26:26:7) (7,-.4) node[fill=white] {$d_3$};
         \draw[] (7.29, -3.27) arc (-24:24:8) (8,-.4) node[fill=white] {$d_4$};
        \draw[very thick] (9.4, -3.42) arc (-20:20:10) (10,-.4) node[fill=white] {$\:\:\:\:\:\:\:\:\:\:\:\:\:\:\:\:{d_5}=\phi_{\text{max}}$};

        \draw[->, very thick, rounded corners=.48cm, red] (9.5,0)--(9.7,-1.5)--(8.2,-1.6)--(8.8,-.35)--(9.4,-2.5)--(9.8,-1)--(8.4,-.2)--(7.9,1.2)--(9.2,1.9)--(9.9,.7)--(8.5,.25)--(7.5,2.5)--(6.5778,2.394);
        \draw[loosely dashed, ->,thick,red] (6.5778,2.394)--(7,0);
        \filldraw[] (9.5,0) circle (2pt);
        \draw[very thin] (9.5,0)--(11.5,1.5) node[fill=white]{$x^0$};

        \draw[->, very thick, rounded corners=.5cm, blue] (7,0)--(6.1,-0.4)--(7.5,-1.2)--(5,-1.4)--(4,-.3)--(2.897777,-0.776457);
        \draw[loosely dashed, ->,thick,blue] (2.897777,-0.776457)--(3,0);

        \draw[->, very thick, rounded corners=.4cm, gr] (3,0)--(1.8,.4)--(2.5,1.4)--(3.2,.5)--(4.2,.06)--(4.3,1.3)--(3.1,2)--(1.9,1)--(1.4,-.7)--(2.5,-.3)--(.9,.9)--(0.453,0.2113);
    \end{tikzpicture}
    \caption{\textit{(Passive agent model)} Depiction of space as considered and notated in the proof of Theorem \ref{thm:singleAgentBound}, with example trajectories for phases/regimes $1$, $2$, and $3$ of the analysis shown in red, blue, and green, respectively.
    Note that this analysis mirrors the general trend in agent motion in the passive agent model (see Figures \ref{fig:exRun} and \ref{fig:expDP}): `fully' random when far away from the cancer site, more biased as approaching, and finally `fully' random once again when extremely close.
    }
    \label{fig:beauty}
\end{figure}

\paragraph{Notable Distances}\label{sec:distances}  We introduce the key distances $d_5,d_4, \dots, d_1$ that govern the behavior of an agent.
In order to do this, we also pick some fixed constant $\delta \in \mathbb{R}_{>0}$ that is a given desired bias that every agent has in moving towards the cancer site; our key distances will be chosen to then guarantee a bias of $\delta$ in the ``blue'' phase (see Figure~\ref{fig:beauty}).\footnote{It might well be that $\delta$ is too large, resulting in an empty interval $[d_2, d_4]$. However, for a broad range of parameters—particularly when the chemical signal is sufficiently strong—the interval remains non-empty.}

\begin{enumerate}
\item $d_5 \coloneqq \phi_{\text{max}}$ is the maximum distance, i.e., $\phi(t)\leq d_5$ always holds.
\item
 $d_4$ represents the \textit{largest} distance from which  the expected progress towards the cancer site in one step $\mathbb{E}[\Delta\phi]$---\eqref{eq:triangle}---is at least $\delta$. In symbols, $d_4 \coloneqq \arg\max_{\phi} \mathbb{E}[\Delta\phi] \geq \delta. $
\item 
 $d_2$ represents the \textit{smallest} distance from which the potential change of \eqref{eq:triangle} is at least $\delta$ throughout $[d_2, d_4]$.
In symbols, $d_2 \coloneqq \arg\min_{\phi} \forall d \in [\phi,d_4]: \E{\Delta \phi(d)} \geq \delta. $
\item $d_3 = \frac{d_2+d_4}{2}$.
 \item  $d_1\coloneqq \epsilon$ is the (maximum) cancer site detection distance, i.e., hitting the ball around $x^*$ of radius $d_1$ is deemed sufficient for having ``reached the cancer site'' by our model definition.
\end{enumerate}

Note that our definition implies $\delta \leq \alpha. $
We also make the following additional assumptions.

\begin{assumptions}\label{assu}
Assume
$\delta \leq \frac{1}{3}$; \quad 
$d_3 - d_2 > \frac{4}{\delta^5}$ 
; 
\quad 
$\frac{d_3^2 + d_3 d_5 \sqrt{2} + d_5^2}{2} \geq 1$; \quad 
$\frac{d_3 - d_2}{\alpha} \geq 1$

\end{assumptions}

We now present the main theorem of this paper, which bounds the time required for a single agent in the passive agent model to reach the cancer site and deliver its treatment.
\begin{theorem}\label{thm:singleAgentBound}
    Working within the passive agent model, consider a single agent $i$.
    Assume $\alpha=\epsilon$ and that Assumptions~\ref{assu} hold.
    Let $\phi(t) \coloneqq || x^{(t)}_i - x^* ||_2$.
    Let $T$ be the random variable denoting the first point in time $t$ for which $\phi(t) \leq \epsilon$, i.e., when the agent reaches the cancer site.
    Let  \begin{equation}\label{eq:s} s^{-1} = \frac{{2}d_1^2}{\pi \alpha{(d_3-d_2)}} \exp\left( -\frac{d_1^2 + d_1 d_2 \sqrt{2} + d_2^2}{\alpha{(d_3-d_2)}} \right) \ \text{and}\ \ s' = 
\frac{\exp\left( -  1/\delta^3\right)}{1-e^{-\delta^3/4}} .
    \end{equation} We have
    $$\mathbb{E}[T] \leq \left(1+\frac{s+1}{s'}\right)\frac{\pi e (d_3^2 + d_3d_5\sqrt{2}+d_5^2)^2}{2\alpha^2 d_3^2} + (s+1) \frac{d_3-d_2}{\delta} + s \frac{d_3-d_2}{\alpha}.
    $$
\end{theorem}

\subsection{Proof Idea for \texorpdfstring{Theorem~\ref{thm:singleAgentBound}}{}}\label{sec:pfsketch}

    Without loss of generality, we let $x^0 = (\phi(0), 0) \in \mathbb{R}^2$.
    By our assumptions, $\phi(0)\leq d_5$. For any $d,d'\in \mathbb{R}$ with $d>d'\geq 0$ we define $T_{d,d'}$ as the random variable denoting the first point in time $t$ for which  $\phi(t) \leq d'$, where $\phi(0) = d$.
    In particular, let $T=T_{\phi(0),d_1}$ be the time the (passive) agent first reaches the target cancer site.
    Since $\phi(0)\leq d_5$, we then have \begin{align}\label{eq:main} \mathbb{E}[T] \leq \mathbb{E}[T_{d_5,d_3}] + \mathbb{E}[T_{d_3,d_2}] + \mathbb{E}[T_{d_2,d_1}]. \end{align}

Each term on the right-hand side can and needs to be bound separately, since the ``drift''\footnote{Here, we are considering the one-dimensional random walk with values in $[0,d_5]$ tracking the agent's current distance $\phi(t)$ from the cancer site $x^*$. So the ``drift'' or ``bias'' in this context then denotes the expected change in the distance to $x^*$ in one step. Note that an unbiased random walk in $\mathbb{R}^2$ has a drift away from the cancer cite, unlike in $\mathbb{R}^1$, where there is no drift.} of the agent is different in each of these three regimes.
We now give an overview of how we treat the analyses of these quantities in the different regimes.
See Figure~\ref{fig:beauty} for an illustration.
Our results hinge on a crucial property: the distance of the agent is majorized by the distance of a two-dimensional unbiased random walk. We prove this in Proposition~\ref{pro:coupling}. 
We now give an overview of the different regimes and their analysis.
\begin{enumerate}
    \item Between \(d_5\) and \(d_4\), the absolute value of drift of the agent is, by definition of \(d_4\), less than $\delta$. Using the fact that the agent is no slower (majorization) than a slightly biased random walk, we can estimate the time it takes to be within the \(d_3\) using  results for unbiased random walks. 
    \item Between \(d_3\) and \(d_2\), the agent is majorized by a $1$D random walk (with respect to $\phi$) with a negative bias of \(-\delta\).
    The reason we don't simply assume the random walk for this phase/regime starts at \(d_4\) is that after just one step outside \(d_4\), it would no longer be a \(|\delta|\)-biased random walk. Thus, we use the regime between \(d_3\) and \(d_2\) and allow the random walk to go all the way to \(d_4\) (which we show will happen very rarely).
    \item Between \(d_2\) and \(d_1\), the bias becomes weak again, and we use our majorization with an unbiased random walk that allows for a clean analysis: We are interested in the time it takes to hit the ball at \((0,0)\) with radius \(d_1 = \epsilon\). Each attempt has a probability of success \(p\). 
    With the remaining probability, the random walk has still failed to hit the ball, but, by construction, it will still be within \(d_3\), and thus, return quickly to \(d_2\). After \(1/p\) trials, we expect the random walk to hit the $d_1$-radius ball.
\end{enumerate}

Proof details appear in Appendix~\ref{appendix A}.

\subsection{Simulation Results for a  Single Passive Agent}
\label{sec: simulation-single-agent}

\begin{figure}
  \centering
  \begin{subcaptiongroup}
    \centering
    \parbox[b]{.48\textwidth}{%
    \centering
    \includegraphics[width=.44\textwidth]{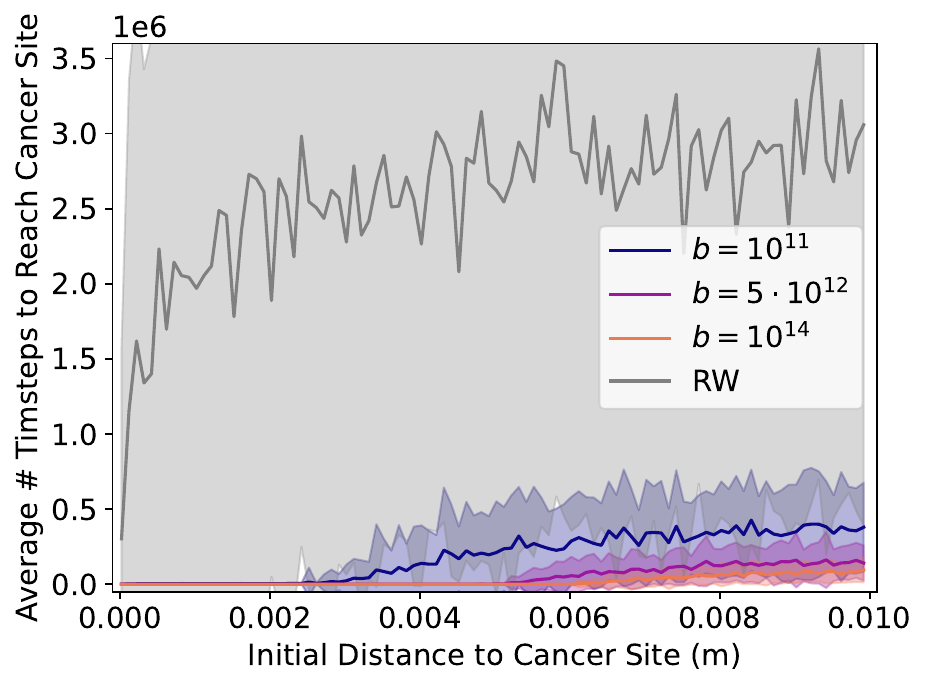}
    \caption{}
    \label{fig:distht}}%
    \hfill
    \parbox[b]{.48\textwidth}{%
    \centering
    \includegraphics[width=.44\textwidth]{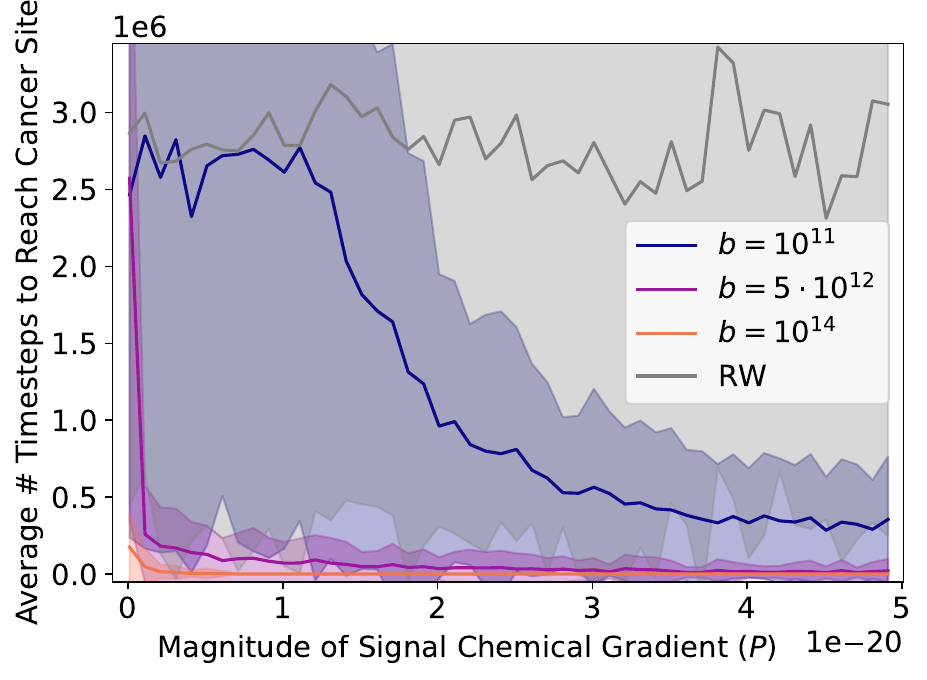}
    \caption{}
    \label{fig:Pht}}%
  \end{subcaptiongroup}
  \caption{\textit{(Passive agent model)} Number of timesteps for a single agent under our passive agent model to reach the cancer site (i.e., hitting time) under varying degrees of orientation-bias, and compared to a simple random walk (``RW'').
  Plotted lines are average hitting times over $100$ trials and shaded regions are standard deviations. (a): How hitting time scales with initial distance to cancer site. Fix signal chemical gradient, $P=10^{-19}$ kg. (b): Effect of the magnitude of the signal chemical gradient on hitting time. Fix initial distance to cancer site at $0.005$ m.
  For both: Fixed parameters include $\alpha=\epsilon=2\cdot 10^{-5}, D=10^{-10}, t^*=10^{4}, \phi_{\text{max}}=0.01$ (SI units). Hitting times were capped at $10^7$ timesteps.
  }
  \label{fig:1agRes}
\end{figure}

Moving beyond analysis, we now present simulation results for the passive agent model.
Since the upper bound presented in Theorem~\ref{thm:singleAgentBound} simply assumes an initial distance of $\phi_{\text{max}}$, the simulation results  provide a more accurate assessment of runtime across a range of varying initial distances.
In actuality, the agents usually reach the cancer site much faster than the bound presented in Theorem~\ref{thm:singleAgentBound}
\footnote{Even when the initial distance is $\phi_{\text{max}}$, the analysis techniques used for Theorem~\ref{thm:singleAgentBound} likely still produce a bound whose lack of tightness is significant; i.e., we expect agents to usually reach the cancer site much faster the analytical bound, which these hitting times being reflected in simulation results.};
The simulation results also better illustrate the effect of certain model parameters on performance.

Still considering only a single agent under the passive agent model with orientation-bias parameter $b$, Figure \ref{fig:1agRes} shows the hitting time results from simulation, plotted in comparison to the simple random walk for a baseline of comparison.
Recall that larger values of $b$ correspond to agents being more likely to orient and move towards the cancer site (i.e., greater amounts of orientation-bias).

In Figure \ref{fig:distht}, without loss of generality---since the signal chemical gradient is radially symmetric---the agent has initial location $x^0 = (\phi(0), 0)$, where $\phi(0)$ is the initial distance between the agent and the cancer site $x^*$.
We fix $P=10^{-19}$ kg.
As $\phi(0)$ increases, hitting times increase, with this decline in performance worsening with less orientation-bias.
However, under these specific environment settings $\gamma(\cdot)$, once within a close enough distance ($\approx 0.004$ m) and thus a steep enough signal chemical gradient, agents reach the cancer site very efficiently regardless of the amount of bias.
Figure \ref{fig:Pht} shows the effect of the magnitude of the signal chemical gradient on hitting times.
We fix $\phi(0)=0.005$ m.
As $P$ increases, hitting time decreases since there is a more noticeable chemical signal farther out from the cancer site, which the agent senses and consequently orients/moves more favorably in expectation.
The effect of $P$ on hitting time is more noticeable under weaker amounts of orientation-bias.
As shown in both subfigures, our movement model significantly outperforms the simple random walk across all settings.

\subsection{Results for Multiple Passive Agents}
\label{sec:multi}
\label{sec: simulation-multi-agents}

We now return to the main application of interest in which at least $75\%$ of the $n$ ($>1$) nanobot agents are tasked with locating the cancer site in a reasonable amount of time such that all of their marginal drug payloads together provide the desired amount of treatment.
We can {easily} obtain the following bound, based on Theorem~\ref{thm:singleAgentBound}.
The proof appears in Appendix~\ref{appA3}.

\begin{corollary}
\label{cor:multi}
Let \( T_1, \dots, T_n \) be i.i.d.\ random variables representing the drug delivery times of \( n \) agents, and let \( t^+ \) be the upper bound derived in Theorem~\ref{thm:singleAgentBound} such that \( \mathbb{E}[T_i] \leq t^+ \ \forall i \in [1, n] \). 
Then, for any number of agents \( n \geq 1 \), at least 75\% of the agents complete delivery by time \( 10t^+ \) with probability at least $1 - 2^{-n/4}$.
\end{corollary}


We complement our theoretical result of Corollary~\ref{cor:multi} with simulations. 
Runtime results are shown in Figure \ref{fig:nAgRes}.
We fix $n=25$.
All $n$ agents have the same initial location $x^0=(\phi(0), 0)$, where $\phi(0)$ is the initial distance between the agents and the cancer site.
Similar to the single agent hitting time results, we see runtimes increasing nonlinearly as the initial distance to the cancer site increases.
Again, agents under the passive agent model significantly outperform the simple random walk as expected (by approximately one to two orders of magnitude, depending on the initial distance).

\begin{figure}
    \centering
    \includegraphics[width=.45\textwidth]{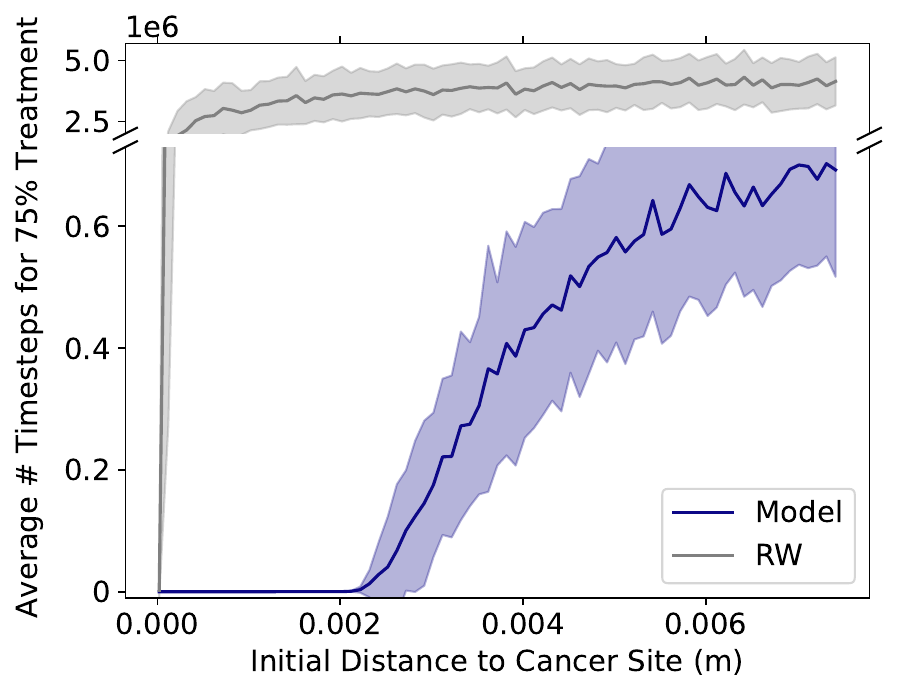}
    \caption{\textit{(Passive agent model)} Runtime results for $25$ agents in the passive agent model for varying initial distances to the cancer site, and in comparison to the pure RW.
    Plotted are average runtimes over $100$ trials and shaded regions are standard deviations.
    Runtimes were capped at $10^7$ timesteps.
    Fixed parameters include $\alpha=\epsilon=2\cdot 10^{-5}, P=10^{-18}, D=10^{-10}, t^*=10^{3}, \phi_{\text{max}}=0.01, b=10^{12}$ (SI units).
    }
    \label{fig:nAgRes}
\end{figure}
Note that we cannot simply extend our single-agent analytical result to the multi-agent case, despite the agents' walks being independent in the passive agent model.
That is, given only the expected time for an individual agent to reach the cancer site (and not knowing more about the distribution), we cannot conclude the expected time for multiple agents
\footnote{This can be demonstrated by a simple toy example with two agents ($n=2$). Here, the time for at least $75\%$ of the agents to reach the cancer site is simply the time for both agents to reach the site. Say the expected hitting time for a single agent (in some given environment/model setting) is $1$. Consider two distinct possibilities that would both correctly yield this expected hitting time of $1$: (a) each agent always takes time $1$, and (b) each agent takes either time $0$ or $2$, each with probability $1/2$.
In case (a), the total expected hitting time for both agents to reach the cancer site is $1$, while in case (b), the total expected hitting time for both agents is $1 1/2$. This proves knowing the expected hitting time for a single agent is not, in general, enough knowledge in order to make a statement on the total expected hitting time of multiple agents.}.

\section{Results for the Active Agent Model}
\label{results2}

Now we present our simulation results for the active agent model.
We first present results in Figure \ref{fig:1to2} on the runtime speedup seen once the very first signal chemical payload is dropped and there is a nonzero chemical gradient to follow.
We then consider the full swarm of $n$ agents and present simulation results in Figure \ref{fig:big} to show the collective benefit of having more total agents.

Beyond agents dropping payloads in order to have \textit{some} chemical signal for others to follow and essentially recreate the environment conditions of the passive agent model, one would hope that there would be some additional collective benefit and runtime improvement in having agents be in control of their own dynamic chemical signal---this is indeed what we observe.
Breaking down this advantageous collective behavior into contributing factors, we begin by looking at the immediate hitting time improvement seen after the very first signal chemical payload is dropped, shown in Figure \ref{fig:1to2}.
The first agent with a chemical signal to follow reaches the cancer site around three orders of magnitude faster than the simple random walk, proving that the time it takes for the very first agent to drop its signal chemical payload is crucial to the overall runtime. 

Figure \ref{fig:big} shows runtime results with different numbers of agents.
The total runtime is effectively the sum of the time it takes for the very first signal chemical payload to be dropped---shown in Figure \ref{fig:nrws}---and the time it takes for the desired $75\%$ of drug payloads to be dropped once the first signal chemical payload is already dropped---shown in Figure \ref{fig:after1sthit}.
The former portion of the total runtime is improved with more agents as the time for the first of $\lfloor n/2 \rfloor$ simple random walks to hit the cancer site trivially decreases in expectation as $n$ increases.
The latter portion of the total runtime is also improved with more agents as the global signal chemical gradient is amplified, resulting in more favorable agent behavior particularly by those that reach the cancer site later on.
Recall Figure \ref{fig:expDP}, which demonstrated that after long enough, the signal chemical gradient flattens out and starts to become useless for agents to follow.
That is, over time, the chemical signal weakens and yields fully random agent motion; however, with more agents, the time between successive dropped payloads decreases which combats this phenomenon, maintaining the chemical gradient at a useful (steep) state, even if it is not amplified.
Ultimately, the active agent model outperforms the simple random walk by approximately three orders of magnitude.
\begin{figure}
    \centering
    \includegraphics[width=0.45\linewidth]{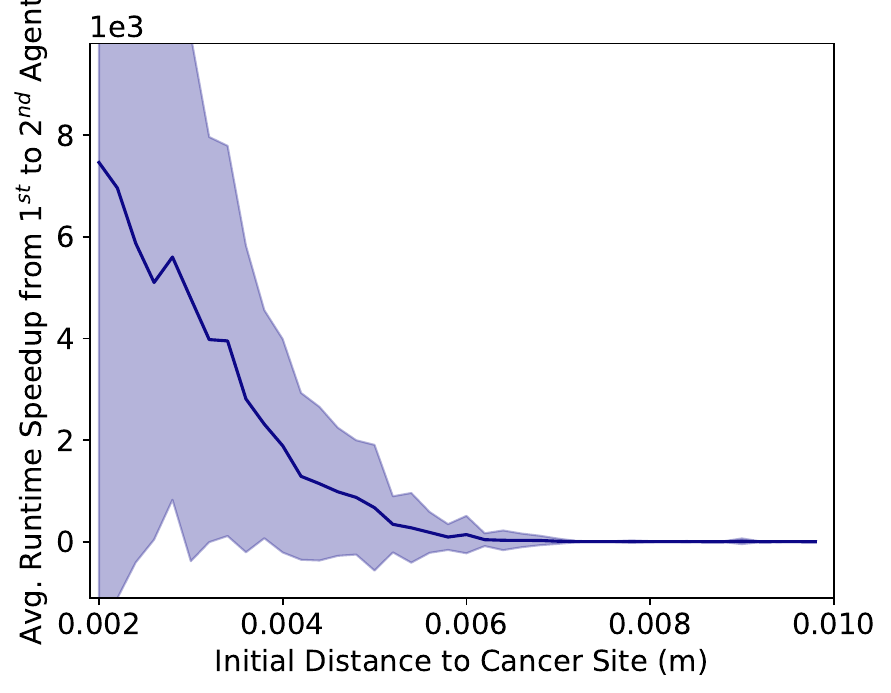}
    \caption{
    \textit{(Active agent model)} Runtime speedup (ratio) from an agent following a simple random walk---``$1^{\text{st}}$ agent''---to the first agent with a nonzero chemical signal to follow in the active agent model---``$2^{\text{nd}}$ agent''.
    Both agents start from the same initial distance.
    The signal chemical gradient for this ``$2^{\text{nd}}$ agent'' is given by $z^{(0)}=1=z^{(t)} \:\forall t$, $t^*_1=0$.
    Fixed parameters include $\alpha=\epsilon=2\cdot 10^{-5}, P=10^{-19}, D=10^{-9}, b=10^{12}, \phi_{\text{max}}=0.01$ (SI units).
    }
    \label{fig:1to2}
\end{figure}
\begin{figure}
    \centering
    \begin{subfigure}[t]{0.48\textwidth}
        \centering
        \includegraphics[width=.9\textwidth]{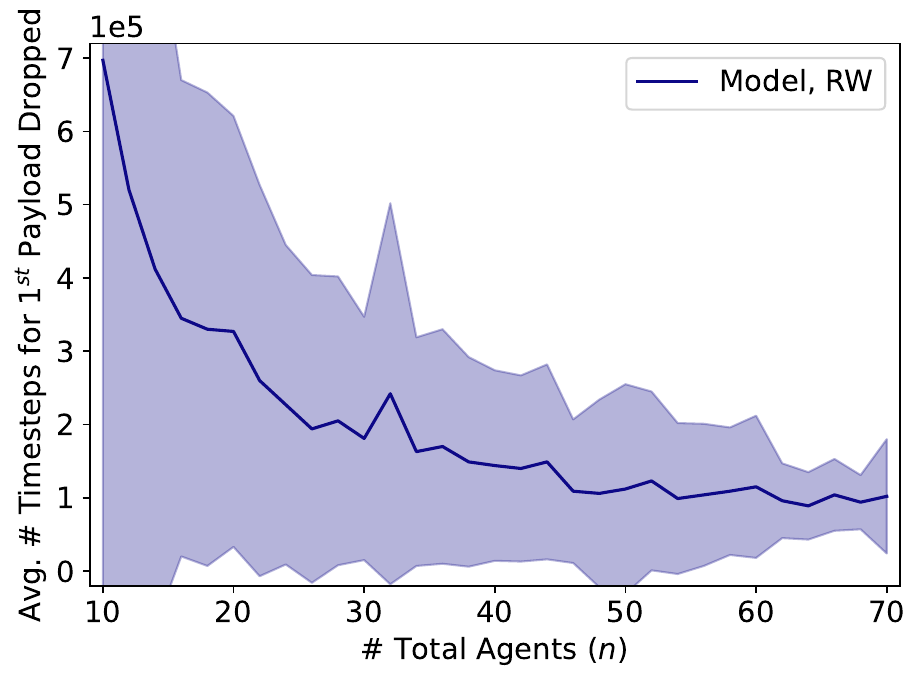}
        \caption{}
        \label{fig:nrws}
    \end{subfigure}%
    ~ 
    \hfill
    \begin{subfigure}[t]{0.48\textwidth}
        \centering
        \includegraphics[width=.9\textwidth]{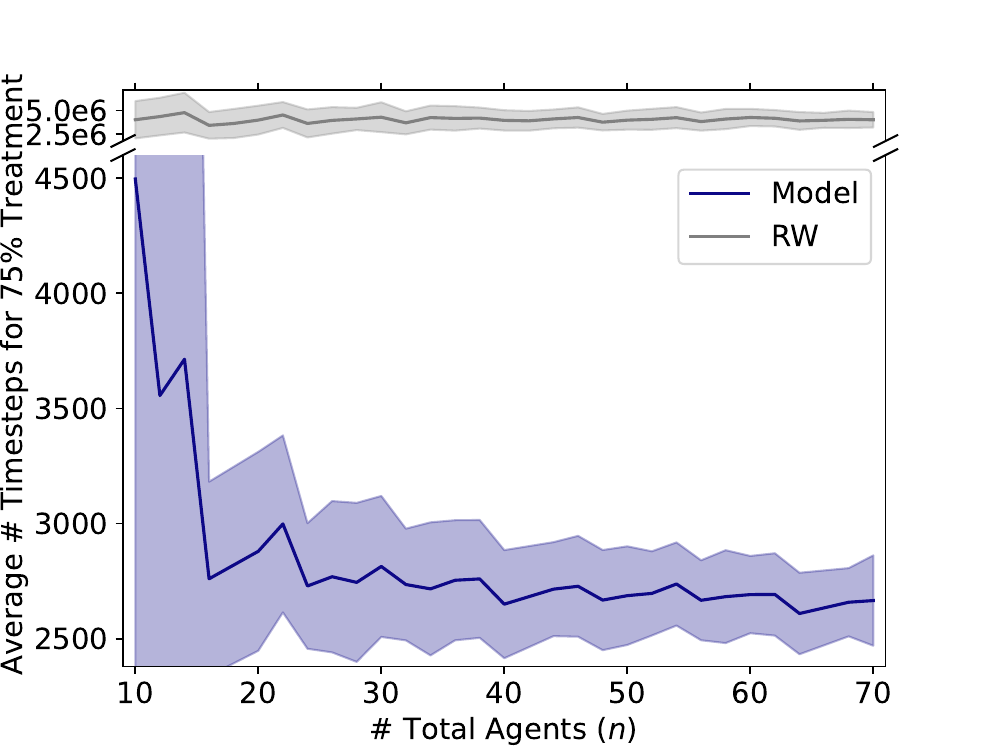}
        \caption{}
        \label{fig:after1sthit}
    \end{subfigure}
    \hfill
    \caption{
    \textit{(Active agent model)} Runtime results for the active agent model, showing the collective benefit of a larger swarm, and in comparison to the simple random walk (``RW''). (a): Time for first signal chemical payload to be dropped with $n$ total agents, equivalent to the time for first of $\lfloor n/2\rfloor$ independent simple random walks to hit cancer site. (b): Time for $75\%$ of $\lceil n/2 \rceil$ drug-carrying agents to reach cancer site once first signal chemical payload is already dropped, for $n$ total agents.
    Note that the combined times of (a) and (b) is exactly the total runtime for $75\%$ of the agents with drug payloads to deliver their treatment starting from the initial site.
    Fixed parameters include $\alpha=\epsilon=2\cdot 10^{-5}, P=10^{-19}, D=10^{-9}, b=10^{12}, \phi_{\text{max}}=0.01$, $\phi(0)=0.005$ (SI units).
    Plotted lines are average runtimes over $100$ trials and shaded regions are standard deviations.
    Runtimes were capped at $10^7$ timesteps.
    }
\label{fig:big}
\end{figure}

\section{Discussion}

We present a formal model of nanobots for cancer detection and treatment, with agent locomotion behavior representative of actual nanoparticles.
Our results demonstrate that for the Orientation-Biased movement model, agents' movement becomes more biased as they get closer to the cancer site, before returning to fully random movement once they are extremely close.
This trend is a result of the specific signal chemical concentration function $\gamma(\cdot)$, not just the movement model in general.

We consider two distinct scenarios, and accompanying variants of the general model: firstly, the passive agent model in which there already exists an endogenous chemical gradient centered at the targeted cancer site; and secondly, the active agent model in which agents must form the chemical signal themselves and maintain it over time by dropping payloads at the cancer site.
We have simulation results for both versions of the model, showing the behavior of both the individual agents as well as the collective swarm.
Across all settings, our model significantly outperforms the fully random walk.
We also provide analytical results for the passive agent model, bounding the time it takes for a single agent to reach the cancer site starting from a given initial distance, as well as the time for $75\%$ of the agents to reach the site with high probability.
Simulation results for the passive agent model show how runtime scales nonlinearly with the initial distance to the cancer site as well as the extent to which a stronger signal chemical gradient benefits performance.
While all walks are independent in the passive agent model, this is not the case in the active agent model.

Simulation results for the active agent model demonstrate a runtime speedup of multiple orders of magnitude with more agents.
With more total agents, the time when the chemical signal first appears is drastically reduced as this coincides with the time when the first of $\lfloor n/2 \rfloor$ simple random walks hits the cancer site.
Further, as more and more agents release their signal chemical payloads over time, the gradient is able to be maintained at a steep---and thus useful---state before fully diffusing/dispersing and becoming uniform---and thus useless.
That is, here, agent movement becomes more biased both when closer to the cancer site \emph{and} as time progresses.
The second agent to reach the cancer site does so faster in expectation than the first agent because of the introduction of a chemical signal to follow, the third agent reaches the site faster in expectation than the second because of an amplified gradient, and so on.

If only in passing, we acknowledge that we did not consider the average clearance time for nanoparticles in the given colloidal environment; for the above mechanism to be indeed useful in practice, the clearance time would need to be sufficiently long such that the later agents have enough time to sense the gradient before being expelled.
Although speculative in the assumptions made on individual agent capabilities, we hope that the results shown for the active agent model are compelling enough to motivate experimentalists when designing and testing new nanoparticles.

\subsection{Future Work}

Future work could vary further environment parameters such as the diffusion coefficient $D$, individual payload size $P$ (in the active agent model), cancer site detection distance $\epsilon$, and maximum boundary distance $\phi_{\text{max}}$, and evaluate the resulting performance of the swarm (under both variants of the model).
Such an exercise, given interesting simulation results, could motivate new experiments to run with actual nanoparticles.
The $3$-dimensional setting also remains to be investigated.
Future work could include studying the effect of different forms of noise on swarm behavior, such as imperfect detection of the cancer site and agent motion impacted by blood flow.

Future work could also consider multiple cancer sites, introducing the problem of how to evenly allocate the agents' treatment across all sites even if certain sites are more easily located.
This setting could call for more involved strategies such as multiple distinct classes of agents, multiple distinct signal chemicals, and different injection times (i.e., different agents introduced into the environment at different points in time).

Lastly, in the active agent model, it is possible that analytical bounds on the expected time for some percentage of the total agents to reach the cancer site can be derived.

\paragraph{Acknowledgments}
Special thanks to Adithya Balachandran, Sabrina Drammis, and Mien Brabeeba Wang for their suggestions that helped lead our earlier research to the presented work.
Also thanks to Sangeeta Bhatia for her suggestions and encouragement.

\bibliographystyle{acm}
\footnotesize{\bibliography{bibliography}}

\begin{thebibliography}{10}

\bibitem{brigger2012nanoparticles}
{\sc Brigger, I., Dubernet, C., and Couvreur, P.}
\newblock Nanoparticles in cancer therapy and diagnosis.
\newblock {\em Advanced drug delivery reviews 64\/} (2012), 24--36.

\bibitem{clementi2021search}
{\sc Clementi, A., d'Amore, F., Giakkoupis, G., and Natale, E.}
\newblock Search via parallel l{\'e}vy walks on z2.
\newblock In {\em Proceedings of the 2021 ACM Symposium on Principles of
  Distributed Computing\/} (2021), pp.~81--91.

\bibitem{crandall1996nanotechnology}
{\sc Crandall, B.}
\newblock {\em Nanotechnology: molecular speculations on global abundance}.
\newblock Mit Press, 1996.

\bibitem{freitas1999nanomedicine}
{\sc Freitas, R.~A.}
\newblock {\em Nanomedicine, volume I: basic capabilities}, vol.~1.
\newblock Landes Bioscience Georgetown, TX, 1999.

\bibitem{fujisawalevy}
{\sc Fujisawa, R., and Dobata, S.}
\newblock Lévy walk enhances efficiency of group foraging in
  pheromone-communicating swarm robots.
\newblock In {\em Proceedings of the 2013 IEEE/SICE International Symposium on
  System Integration\/} (2013), pp.~808--813.

\bibitem{ghosh2017stimuli}
{\sc Ghosh, A., Yoon, C., Ongaro, F., Scheggi, S., Selaru, F.~M., Misra, S.,
  and Gracias, D.~H.}
\newblock Stimuli-responsive soft untethered grippers for drug delivery and
  robotic surgery.
\newblock {\em Frontiers in Mechanical Engineering 3\/} (2017).

\bibitem{golestanian2005propulsion}
{\sc Golestanian, R., Liverpool, T.~B., and Ajdari, A.}
\newblock Propulsion of a molecular machine by asymmetric distribution of
  reaction products.
\newblock {\em Physical review letters 94}, 22 (2005), 220801.

\bibitem{gomez2021markov}
{\sc G{\'o}mez, J.~T., Wendt, R., Kuestner, A., Pitke, K., Stratmann, L., and
  Dressler, F.}
\newblock Markov model for the flow of nanobots in the human circulatory
  system.
\newblock In {\em Proceedings of the Eight Annual ACM International Conference
  on Nanoscale Computing and Communication\/} (2021), pp.~1--7.

\bibitem{gwisai2022magnetic}
{\sc Gwisai, T., Mirkhani, N., Christiansen, M.~G., Nguyen, T.~T., Ling, V.,
  and Schuerle, S.}
\newblock Magnetic torque--driven living microrobots for increased tumor
  infiltration.
\newblock {\em Science Robotics 7}, 71 (2022), eabo0665.

\bibitem{howse2007self}
{\sc Howse, J.~R., Jones, R.~A., Ryan, A.~J., Gough, T., Vafabakhsh, R., and
  Golestanian, R.}
\newblock Self-motile colloidal particles: from directed propulsion to random
  walk.
\newblock {\em Physical review letters 99}, 4 (2007), 048102.

\bibitem{claudiavesicles}
{\sc Joseph, A., Contini, C., Cecchin, D., Nyberg, S., Ruiz-Perez, L.,
  Gaitzsch, J., Fullstone, G., , Tian, X., Azizi, J., Preston, J., Volpe, G.,
  and Battaglia, G.}
\newblock Chemotactic synthetic vesicles: Design and applications in
  blood-brain barrier crossing.
\newblock {\em Science Advances\/} (2017).

\bibitem{kostarelos2010nanorobots}
{\sc Kostarelos, K.}
\newblock Nanorobots for medicine: how close are we?
\newblock {\em Nanomedicine 5}, 3 (2010), 341--342.

\bibitem{LENGLER_STEGER_2018}
{\sc Lengler, J., and Steger, A.}
\newblock Drift analysis and evolutionary algorithms revisited.
\newblock {\em Combinatorics, Probability and Computing 27}, 4 (2018),
  643–666.

\bibitem{martel2009flagellated}
{\sc Martel, S., Mohammadi, M., Felfoul, O., Lu, Z., and Pouponneau, P.}
\newblock Flagellated magnetotactic bacteria as controlled mri-trackable
  propulsion and steering systems for medical nanorobots operating in the human
  microvasculature.
\newblock {\em The International journal of robotics research 28}, 4 (2009),
  571--582.

\bibitem{reynoldslevy}
{\sc Reynolds, A.~M., and Rhodes, C.~J.}
\newblock The lévy flight paradigm: random search patterns and mechanisms.
\newblock {\em Ecology 90}, 4 (2009), 877--887.

\bibitem{sanchez_catalase}
{\sc Serra-Casablancas, M., Di~Carlo, V., Esporrín-Ubieto, D., Prado-Morales,
  C., Bakenecker, A.~C., and Sánchez, S.}
\newblock Catalase-powered nanobots for overcoming the mucus barrier.
\newblock {\em ACS Nano 18}, 26 (2024), 16701--16714.

\bibitem{diffusion}
{\sc Shi, J.}
\newblock Diffusion of point source and biological dispersal, {N}otes.
\newblock William \& Mary, Math 490-01 Partial Differential Equations and
  Mathematical Biology, 2006.

\bibitem{doi:10.1073/pnas.0610298104}
{\sc Simberg, D., Duza, T., Park, J.~H., Essler, M., Pilch, J., Zhang, L.,
  Derfus, A.~M., Yang, M., Hoffman, R.~M., Bhatia, S., Sailor, M.~J., and
  Ruoslahti, E.}
\newblock Biomimetic amplification of nanoparticle homing to tumors.
\newblock {\em Proceedings of the National Academy of Sciences 104}, 3 (2007),
  932--936.

\bibitem{doi.org/10.1038/nmat3049}
{\sc von Maltzahn, G., Park, J.-H., Lin, K.~Y., Singh, N., Schwöppe, C.,
  Mesters, R., Berdel, W.~E., Ruoslahti, E., Sailor, M.~J., and Bhatia, S.~N.}
\newblock Nanoparticles that communicate in vivo to amplify tumour targeting.
\newblock {\em Nature Materials 10}, 7 (2011), 1476--4660.

\bibitem{sanchez_lm}
{\sc Xu, D., Hu, J., Pan, X., Sánchez, S., Yan, X., and Ma, X.}
\newblock Enzyme-powered liquid metal nanobots endowed with multiple biomedical
  functions.
\newblock {\em ACS Nano 15}, 7 (2021), 11543--11554.

\bibitem{zhang2023advanced}
{\sc Zhang, D., Gorochowski, T.~E., Marucci, L., Lee, H.-T., Gil, B., Li, B.,
  Hauert, S., and Yeatman, E.}
\newblock Advanced medical micro-robotics for early diagnosis and therapeutic
  interventions.
\newblock {\em Frontiers in Robotics and AI 9\/} (2023), 1086043.

\end{thebibliography}

\clearpage

\appendix
\section{Supplementary Information}
\subsection{Proof of \texorpdfstring{Theorem~\ref{thm:singleAgentBound}}{}}
\label{appendix A}
We now argue that regardless of the given phase/regime, our passive agent is no slower than Brownian motion with an equivalent step size of $\alpha$.
Recalling that $\beta$ is the scalar orientation of an agent (see Figure~\ref{fig:triangleOnestep}), we consider two stochastic processes with the following probability density functions (PDF's) for $\beta$:
\begin{itemize}
    \item \textit{Passive Agent (\agent):} Steps of length $\alpha$ are taken with a bias 
    towards the target, modeled by a Gaussian distribution:
    \[
    p_{\text{\agent}}(\beta) = \frac{1}{\sqrt{2\pi} \sigma} e^{ -\frac{\beta^2}{2\sigma^2}} 
    \]
    
    \item \textit{Unbiased Random Walk (RW):} Steps of length $\alpha$ are taken uniformly in all directions:
    \[
    p_{\text{RW}}(\beta) = \frac{1}{2\pi} 
    \]
\end{itemize}
We define the decrease in potential as a function of the angle, implicitly fixing $t$:
\[
\psi(\beta) = \psi(\beta,t)= \phi(t) - \sqrt{\phi(t)^2 - 2\phi(t)\alpha \cos(\beta) + \alpha^2}
\]





\begin{lemma}\label{lem:psi}
We have that $\psi(\beta)$ is monotonically decreasing for $\beta \in [0, \pi]$.
\end{lemma}
\begin{proof}

Taking the derivative with respect to \( \beta \):

\[
\frac{d \psi(\beta)}{d \beta} = - \frac{ \phi(t) \alpha \sin(\beta) }{ \sqrt{\phi(t)^2 - 2\phi(t)\alpha \cos(\beta) + \alpha^2} } \leq 
- \frac{ \phi(t) \alpha \sin(\beta) }{ \sqrt{(\phi(t)+\alpha)^2 } } .
\]
Since \( \sin(\beta) \geq 0 \) for \( \beta \in [0, \pi] \) and all other terms are positive, we have:
\[
\frac{d \psi(\beta)}{d \beta} \leq 0
\]
Thus, \( \psi(\beta) \) is monotonically decreasing.

\end{proof}

The following proposition allows us to use the stochastic dominance. 
\begin{proposition}\label{pro:coupling}
Let \( X_{\text{\agent}}(t) \) and \( X_{\text{RW}}(t) \) be  location vectors corresponding to the Passive Agent (\agent) and Unbiased Random Walk (RW) processes, respectively.
There exists a coupling such that for all $t$, \[\Pr(\phi(X_{\text{\agent}}(t))\leq \phi(X_{\text{RW}}(t)))=1.\]
\end{proposition}
\begin{proof}

We say that
$
Y_{\text{RW}}(t) \geq_{\text{st}} Y_{\text{\agent}}(t)
$
to indicate that \emph{stochastic dominance} holds between the two random variables. 
Formally, stochastic dominance means that for all real numbers \( x \),
$
\mathbb{P}(Y_{\text{RW}}(t) \geq x) \geq \mathbb{P}(Y_{\text{\agent}}(t) \geq x).
$
We first show that the Unbiased Random Walk (RW) process stochastically dominates the Passive Agent (\agent) process for one step.
We write
$Y_{\text{RW}}(t) = \phi(X_{\text{RW}}(t))$ and
$Y_{\text{\agent}}(t) = \phi(X_{\text{\agent}}(t))$.
We thus want to prove   for all \( x \), and $Y_{\text{\agent}}(t-1)\leq Y_{\text{RW}}(t-1)$,
\[ \Pr( Y_{RW}(t)  \geq x ~|~ Y_{RW}(t-1) ) \geq \Pr( Y_{\agent}(t)  \geq x ~|~ Y_{\agent}(t-1) ).   \]
Let $\beta_{\text{RW}}$ and $\beta_{\text{\agent}}$ be the $\beta$'s drawn from the processes in the current time round.
We have that 
\[ Y_{\text{\agent}}(t) = - 
 \psi(\beta_{\text{\agent}}) + Y_{\text{\agent}}(t-1), \:\:\: Y_{\text{RW}}(t) = -\psi(\beta_{\text{RW}})+ Y_{\text{RW}}(t-1).
\]
%
%
Since \( \psi(\cdot) \) is monotonically decreasing by Lemma~\ref{lem:psi}, the mapping \( \beta \mapsto -\psi(\beta) \) is increasing for $\beta \in [0, \pi]$. 
Thus, to prove the one-step stochastic domination, it is sufficient to show that the distribution of \( \beta_{\text{RW}} \sim p_{\text{RW}} \) stochastically dominates that of \( \beta_{\text{\agent}} \sim p_{\text{\agent}} \) , that is:
\[
\Pr( \beta_{\text{RW}} \geq x ) \geq \Pr( \beta_{\text{\agent}} \geq x ) \quad \text{for all } x.
\]
Note that $F_{\text{\agent}}(\beta) = \int_{-\beta}^{\beta} p_{\text{\agent}}(u) \, du = 2\int_{0}^{\beta} p_{\text{\agent}}(u) \, du $, due to symmetry.
This follows directly from the fact that the cumulative distribution functions (CDF's) satisfy
\[
F_{\text{\agent}}(\beta) = 2\int_{0}^{\beta} p_{\text{\agent}}(u) \, du \geq 2\int_{0}^{\beta} p_{\text{RW}}(u) \, du = F_{\text{RW}}(\beta)
\quad \text{for all } \beta \in [0, \pi].
\]
Since the density \( p_{\text{\agent}}(\beta) \) is monotonically decreasing and must integrate to 1 over \( [0, \pi] \), there exists a threshold \( \beta^* \) such that \( p_{\text{\agent}}(\beta) \geq p_{\text{RW}}(\beta) \) for \( \beta < \beta^* \) and \( p_{\text{\agent}}(\beta) \leq p_{\text{RW}}(\beta) \) for \( \beta > \beta^* \), which implies that the CDF \( F_{\text{\agent}}(\beta) \geq F_{\text{RW}}(\beta) \) for all \( \beta \in [0, \pi] \).
Therefore, the RW process stochastically dominates the \agent process in a single step. By induction and the Markov property, this extends to all \( t \):
\[
Y_{\text{RW}}(t) \geq_{\text{st}} Y_{\text{\agent}}(t).
\]
Finally, by Strassen’s theorem, this stochastic domination implies the existence of a coupling such that
\[
\mathbb{P}(Y_{\text{\agent}}(t) \leq Y_{\text{RW}}(t)) = 1.
\]

\end{proof}

 We now proceed by bounding the different `hitting times' of each regime.

    
\begin{lemma}\label{lem:53}
It holds that
\begin{equation*}
    \E{T_{d_5,d_3}} \leq \frac{\pi e (d_3^2 + d_3d_5\sqrt{2}+d_5^2)^2}{2\alpha^2 d_3^2}.
\end{equation*}
\end{lemma}
\begin{proof}
Consider the red trajectory in Figure \ref{fig:beauty}, though here we assume $\phi(0)=d_5$ as we are analyzing $T_{d_5,d_3}$.
We use the stochastic dominance and analyze  standard Brownian motion with fixed step size $\alpha$, and upper bound the expected time it takes such a random walk to hit the radial area (circle) around the cancer site $x^*$ of radius $d_3$.
This will be an upper bound on $\E{T_{d_5,d_3}}$.
Define the (independent) random variables (coordinate) $X_t$ and $Y_t$ such that $(X_t, Y_t) = x_i^{(t)}$.

We consider \textit{rounds}, each containing $L$ (to be defined) timesteps.
At the end of each round, we reset $X_L$ and $Y_L$ to be $X_0=d_5$ and $Y_0 = 0$, respectively.
To be clear, $\phi(t)\coloneqq \sqrt{X_t^2 + Y_t^2}$---i.e., our convention for consistency is to also reset the current timestep $t$ to be zero at the end of each round. 
Following Brownian motion with zero drift and step size $\alpha$, we have $X_L \sim \mathcal{N}(d_5, \sigma_L^2)$ and $Y_L \sim \mathcal{N}(0,\sigma_L^2)$, where $\sigma_L^2= \alpha^2 L/2$.

Let $\mathcal{E}_3$ be the event that the agent successfully hits the ball around $x^*$ of radius $d_3$ during the span of a given individual round.
The probability of hitting the ball at some point \textit{during} a given round is greater than or equal the probability of being within the ball at the end of the same given round, trivially, yielding $\Pr(\mathcal{E}_3) \geq \Pr(\phi(L) \leq d_3)$.
Considering the largest square inscribed in the ball of radius $d_3$, we have
\begin{align}
    \Pr(\mathcal{E}_3) &\geq \Pr(\phi(L) \leq d_3) \geq \Pr\left(X_L \in \left[-\frac{d_3 \sqrt{2}}{2}, \frac{d_3 \sqrt{2}}{2} \right]\right) \nonumber\\
    &\hspace{140pt} \cdot \Pr\left(Y_L \in \left[-\frac{d_3 \sqrt{2}}{2}, \frac{d_3 \sqrt{2}}{2} \right]\right) \nonumber\\
    &\geq \left( d_3 \sqrt{2} \cdot f_{X_L}\left(-\frac{d_3 \sqrt{2}}{2}\right) \right) \left( d_3 \sqrt{2} \cdot f_{Y_L}\left(\frac{d_3 \sqrt{2}}{2}\right) \right)\nonumber\\
    &= 2 d_3^2 \cdot \frac{1}{\sqrt{2\pi \sigma_L^2}} \exp\left(\frac{-(-\frac{d_3 \sqrt{2}}{2} - d_5)^2}{2\sigma_L^2} \right) \cdot \frac{1}{\sqrt{2\pi \sigma_L^2}} \exp\left(\frac{-(\frac{d_3 \sqrt{2}}{2})^2}{2\sigma_L^2} \right) \nonumber\\
    &= \frac{d_3^2}{\pi \sigma_L^2} \exp\left( \frac{-(\frac{d_3 \sqrt{2}}{2} + d_5)^2 - (\frac{d_3 \sqrt{2}}{2})^2}{2\sigma_L^2} \right) \nonumber\\
    &= \frac{d_3^2}{\pi \sigma_L^2} \exp\left( -\frac{d_3^2 + d_3 d_5 \sqrt{2} + d_5^2}{2\sigma_L^2} \right) \nonumber
\end{align}
where $f_{X_L}$ and $f_{Y_L}$ are the probability density functions of $X_L$ and $Y_L$, respectively.
We choose $L = (d_3^2 + d_3d_5\sqrt{2}+d_5^2)/\alpha^2$.

Ultimately we have 
\begin{align*}
    \E{T_{d_5,d_3}} &= \mathbb{E}[\text{\# timesteps to first hit ball of radius $d_3$}] \\
    &\leq \mathbb{E}[\text{\# rounds to first hit}] \cdot L = \frac{L}{\Pr(\mathcal{E}_3)} \\
    &=  \frac{\pi e L \sigma_L^2}{d_3^2}  = \frac{\pi e \alpha^2 L^2}{2d_3^2} = \frac{\pi e (d_3^2 + d_3d_5\sqrt{2}+d_5^2)^2}{2\alpha^2 d_3^2}
\end{align*}

\end{proof}

To study the next phase, i.e, $\E{T_{d_3,d_2}}$, we will reduce the problem to  decreasing the potential $\phi(\cdot)$ directly---working with the one dimensional potential function is possible due to the drift in this phase.
Let $(X_t)_{t\geq 0}$ denote the steps taken by the passive agent ($i$), i.e., $X_t=\phi\left(x_i^{(t)}\right) \in \mathbb{R}$. 

\begin{lemma}\label{lem:32}
It holds that
\begin{equation*}
    \mathbb{E}[T_{d_3,d_2} ]\leq  \frac{d_3 -d_2}{\delta} + \frac{1}{p} \mathbb{E}[T_{d_5,d_3}].
\end{equation*}
\end{lemma}

\begin{proof}
In order to prove this result, we start with defining a particular instance of a one-dimensional, biased random walk which we call the \emph{Idealized Random Walk} and express with $Y$. Concretely, $Y$ is defined as Brownian motion, with step size $\alpha$ and probabilities adjusted so that in expectation $\E{\Delta(\phi(Y))}=\delta$ regardless of the current position.
In particular, it behaves as if the red area/regime does not exist; that is, beyond the green area, it always has a bias of $\delta$.

Let \( T^* \) denote the first point in time when the \emph{Idealized Random Walk}, starting at the $d_3$ radius, hits the $d_2$ radius. 
We denote the steps of the idealized random walk by $(Y_t)_{t\geq 0}$ with $Y_0=d_3$.
Let $p$ denote the probability of the \emph{Idealized Random Walk} first reaching $d_4$ before $d_2$.
Let \( Z \) represent the number of times the \emph{Idealized Random Walk} $(Y_t)_{t
\geq 0}$ leaves the area $[d_2,d_4]$. 
When this occurs, we simply restart the process and examine a new idealized random walk starting at \( d_3 \).
We now argue that we have the bound  \begin{equation}\label{eq:major} \E{T_{d_3,d_2}} \leq \E{T^*} + \E{Z}\cdot \E{T_{d_5, d_3}}. 
\end{equation}

To do this, consider the following process. Whenever there is a timestep $\tau$ where the passive agent $(X_t)_{t\geq 0}$ enters the area $[d_4,d_5]$, that is, when $X_\tau \geq d_4$, then so does the idealized one, that is, $Y_\tau = d_4$.
The time it takes to return to the area $[d_2,d_3]$ is clearly upper bounded by $T_{d_5,d_3}$, i.e., the time for the passive agent starting at the maximum boundary. Whenever the passive agent is back in the area $[d_2,d_3]$ we can use a coupling\footnote{Such a coupling can be constructed similarly to Proposition~\ref{pro:coupling}. } with the idealized one to derive an upper bound, and for every time it enters the area $[d_4,d_5]$ we upper bound the required time with $T_{d_5,d_3}$ yielding \eqref{eq:major}.

When $Z=0$, the distribution of $T^*$ for the \emph{Idealized Random Walk} stochastically dominates \( T_{d_3,d_2} \), since the idealized walk has a drift of exactly $\delta$ and the original walk has a drift of at least $\delta$.\footnote{There are many walks that have the same drift, we select one that changes $p_{\text{\agent}}$ so that the expected potential drop is exactly $\delta$, and independent of $\sigma$.}
We can subsequently apply Lemma~\ref{lem:progress}, with 
$y=d_2-d_3<0$ and $z=d_4-d_3>0$; note that we have $|y|=|z|$ since $d_3$ is chosen such that $d_2$ and $d_4$ are equidistant. 

Now, whenever the \emph{Idealized Random Walk} reaches the starting position $d_3$ we can apply the result from Lemma~\ref{lem:progress} to derive an upper bound on the probability $p$ of reaching $d_4$ before reaching $d_2$. 
Note that by   Assumptions~\ref{assu} we have that  $d_3 - d_2 = d_4 - d_3 = |y| = |z| > 4 / \delta^5$ and thus  $p = \exp\left( -  \delta^2 z/4\right)
\frac{1}{1-e^{-\delta^3/4}} \leq 
\exp\left( -  1/\delta^3\right)
\frac{1}{1-e^{-\delta^3/4}} 
< 1$ as desired.
We have \( \mathbb{E}[Z] = \frac{1 - p}{p} \). 
Thus, by \eqref{eq:major} and the geometric series, we have:
\begin{align*}
\mathbb{E}[T_{d_3,d_2}] &= \sum_{i=0}^\infty \mathbb{E}[T_{d_3,d_2} \mid Z = i] \Pr(Z = i) \\
&\leq \sum_{i=0}^\infty \left(\mathbb{E}[T^*] + i \cdot \mathbb{E}[T_{d_5,d_3}]\right) \Pr(Z = i) \\
&= \mathbb{E}[T^*] + \mathbb{E}[Z] \cdot \mathbb{E}[T_{d_5,d_3}] \\
&= \mathbb{E}[T^*] + \frac{1 - p}{p} \mathbb{E}[T_{d_5,d_3}]
\end{align*}

It remains to prove 
\begin{align*}
     \E{T^*} &\leq  \frac{d_3 -d_2}{\delta}.
\end{align*}

This follows by the definition of $d_2,d_3$ and Lemma~\ref{thm:drift} using the identity function $g(x)=x$.

\end{proof}



\begin{lemma}\label{lem:21}
Recall the definition of $s$ in \eqref{eq:s}. We have 
 \[ \E{T_{d_2,d_1}} \leq s \left(\frac{d_3-d_2}{\alpha} + \mathbb{E}[T_{d_3,d_2}]\right). \]
\end{lemma}

\begin{proof}
    Consider the green trajectory in Figure \ref{fig:beauty}.
    Building on Proposition~\ref{pro:coupling}, we use the stochastic dominance between the ($2$D) unbiased random walk/standard Brownian motion with step length $\alpha$ (denoted by RW)  and our process. 
   We now upper bound the expected time it takes for RW to hit the disc around $x^*$ of radius $d_1=\epsilon$, starting at a distance $d_2$ from $x^*$.
    Note that the ``boundary conditions'' are different from the analogous analysis of $\E{T_{d_5,d_3}}$ as the agent can leave the ball of radius $d_2$ indefinitely at any point.

    We analyze the subproblem in which, without loss of generality (up to a rotation of the plane), $x_i^{(0)}=(d_2,0)$, and we are seeking an upper bound on the time it takes for $\phi(t)\leq d_1=\epsilon$ to first hold.
    Define the (independent) random variables $X_t$ and $Y_t$ such that $(X_t,Y_t)=x_i^{(t)}$.
    As in our analysis of $\E{T_{d_5,d_3}}$, we consider \textit{rounds}, each containing $L'$ (to be defined) timesteps.
    At the end of each round, we reset $X_{L'}$ and $Y_{L'}$ to be $X_0 = d_2$ and $Y_0=0$, respectively.
    To be clear, $\phi(t)\coloneqq \sqrt{X_t^2 + Y_t^2}$---i.e., our convention for consistency is also to reset the current timestep $t$ to be zero at the end of each round.
    Following Brownian motion with zero drift and step length $\alpha$, we have $X_{L'}\sim \mathcal{N}(d_2,\sigma_{L'}^2)$ and $Y_{L'}\sim \mathcal{N}(0,\sigma_{L'}^2)$ with $\sigma_{L'}^2=\alpha^2L'/2$.

    Let $\mathcal{E}_1$ be the event that the agent, starting from an initial distance of $d_2$, successfully hits the ball around $x^*$ of radius $d_1 = \epsilon$ during the span of a given individual round.
    Considering the largest square inscribed in the ball of radius $d_1$, we have
    \begin{align*}
        \Pr(\mathcal{E}_1) &\geq \Pr(\phi(L') \leq d_1) \geq \Pr\left(X_{L'} \in \left[-\frac{d_1 \sqrt{2}}{2}, \frac{d_1 \sqrt{2}}{2} \right]\right) \\
        &\hspace{140pt} \cdot \Pr\left(Y_{L'} \in \left[-\frac{d_1 \sqrt{2}}{2}, \frac{d_1 \sqrt{2}}{2} \right]\right) \\
        &\geq \left( d_1 \sqrt{2} \cdot f_{X_{L'}}\left(-\frac{d_1 \sqrt{2}}{2}\right) \right) \left( d_1 \sqrt{2} \cdot f_{Y_{L'}}\left(\frac{d_1 \sqrt{2}}{2}\right) \right)
    \end{align*}
    where $f_{X_{L'}}$ and $f_{Y_{L'}}$ are the probability density functions of $X_{L'}$ and $Y_{L'}$, respectively.
    Recall we have $X_{L'}\sim \mathcal{N}(d_2,\sigma_{L'})$ and $Y_{L'}\sim \mathcal{N}(0,\sigma_{L'})$, yielding
    \begin{align*}
        \Pr(\mathcal{E}_1) \geq \Pr(\phi(L') \leq d_1) &\geq 2 d_1^2 \cdot \frac{1}{\sqrt{2\pi \sigma_{L'}^2}} \exp\left(\frac{-(-\frac{d_1 \sqrt{2}}{2} - d_2)^2}{2\sigma_{L'}^2} \right) \nonumber\\
        &\hspace{110pt} \cdot \frac{1}{\sqrt{2\pi \sigma_{L'}^2}} \exp\left(\frac{-(\frac{d_1 \sqrt{2}}{2})^2}{2\sigma_{L'}^2} \right)\\
        &= \frac{d_1^2}{\pi \sigma_{L'}^2} \exp\left( \frac{-(\frac{d_1 \sqrt{2}}{2} + d_2)^2 - (\frac{d_1 \sqrt{2}}{2})^2}{2\sigma_{L'}^2} \right)\\
        &= \frac{d_1^2}{\pi \sigma_{L'}^2} \exp\left( -\frac{d_1^2 + d_1 d_2 \sqrt{2} + d_2^2}{2\sigma_{L'}^2} \right) 
    \end{align*}
    
    One portion of the total runtime for this phase/regime still clearly unaccounted for within this regime is the time required for an agent to return to a distance of $d_2$ if they happen to travel farther away during a given round.
    We will assert that $L' \leq \frac{d_3 - d_2}{\alpha}$, such that an agent can never travel farther than $d_3$ from $x^*$ within a single round (starting at $d_2$).
    Then this expected cost to return to $d_2$ after a round is bounded above by $\E{T_{d_3,d_2}}$.
    We simply set $L'=\frac{d_3 - d_2}{\alpha}$.

   Each round incurs the cost of the round's allocated time itself in addition to the cost of the time needed, if any, to get back to a distance of $d_2$ to start the next round.
   Let $N\sim Geom(p)$ denote the number of trials, with $\E{N}= 1/\Pr(\mathcal{E}_1)$. We have that the total expected time conditioning on $N$ is
   $(N-1)(L'+ \E{T_{d_3,d_2}})+L'=N(L'+ \E{T_{d_3,d_2}})-\E{T_{d_3,d_2}}\leq
   N(L'+ \E{T_{d_3,d_2}})
   $.
   So in total, we have 
   \begin{align*}
   \E{T_{d_2,d_1}} &\leq  \frac{1}{\Pr(\mathcal{E}_1)}(L'+ \E{T_{d_3,d_2}}). 
   \end{align*}

\end{proof}

The proof of Theorem~\ref{thm:singleAgentBound} follows from Lemmas~\ref{lem:53}, \ref{lem:32}, and \ref{lem:21}.

\paragraph{Proof of \texorpdfstring{Theorem~\ref{thm:singleAgentBound}}{}}
\begin{proof}
Putting everything together now, we get
    \begin{align*}
    \E{T} &\leq \E{T_{d_5,d_3}} + \E{T_{d_3,d_2}} + \E{T_{d_2,d_1}}\\
    &\leq \E{T_{d_5,d_3}} + \E{T_{d_3,d_2}}+ s \left(\frac{d_3-d_2}{\alpha} + \mathbb{E}[T_{d_3,d_2}]\right)\\
&\leq \E{T_{d_5,d_3}} + (s+1) \mathbb{E}[T_{d_3,d_2}] + s \frac{d_3-d_2}{\alpha}\\
&\leq  \E{T_{d_5,d_3}} + (s+1) \left( \frac{d_3 -d_2}{\delta} + \frac{1}{s'} \mathbb{E}[T_{d_5,d_3}]\right) + s \frac{d_3-d_2}{\alpha}\\
&= 
\left(1+\frac{s+1}{s'}\right)\E{T_{d_5,d_3}} + (s+1) \frac{d_3-d_2}{\delta} + s \frac{d_3-d_2}{\alpha}
    \end{align*}

This concludes the proof.
\end{proof}

\subsection{Auxiliary Lemmas}

The following lemma is used in the proof of Lemma~\ref{lem:32} to bound the hitting time of the agent from $d_3$ to $d_2$. 
\begin{lemma}\label{lem:progress}
Consider a random $X_t$, which satisfies $-1\leq X_t-X_{t-1}\leq 1$ and $\E{X_t-X_{t-1}} \leq -\delta$ , where $\delta > 0$.
Let $X_0=0$ be the starting value. 
Fix some values $y,z$ arbitrarily with $y<0$ and $z\geq 4$.
Let $T$ be the first time $X_t$ reaches either $y$ or $z$: $T=\min \{ t~|~X_t = y\text{ or } X_t = z\}$.
Then
\[ \Pr(X_T = z) \leq \exp\left( -  \delta^2 z/4\right)
\frac{1}{1-e^{-\delta^3/4}} \]
\end{lemma}

\begin{proof}
Define the interval lengths $t_i = \frac{z}{2} + i \frac{\delta}{2}$ for all integers $i\geq 1$ and $t_0 = 0$.
Letting $\tau_i = \sum_j^i t_j$, we have $\E{X_{\tau_i}} \leq -\delta \sum_j^i t_j$ for all $i\geq 0$.
Let $\mathcal{E}_i$ for $i\geq 1$ be the event that $X_{\tau_{i}}-X_{\tau_{i-1}}\leq -\frac{\delta t_i}{2}$.

Observe that if we condition on $\mathcal{E}_{i-1}, \dots, \mathcal{E}_1$, then the random walk is at a position $X_{\tau_{i}} \leq - \delta\sum_j^{i-1} t_{j}$.  
Thus, even in the worst case in which the next $t_i$ timesteps all move the random walk one to the right (increase by one), the random walk will still not hit $z$ during the $i^{\text{th}}$ interval---i.e., $X_t<z$ for $t\in[\tau_{i-1}, \tau_{i})= [\sum_j^{i-1} t_j,\sum_j^i t_{j})$---since using $z\geq 4$ we have $X_{\tau_i} \leq X_{\tau_{i-1}} + t_i \leq  -\delta \sum_j^{i-1} t_j + t_i \leq -\delta(i-1)\frac{z}{2}  +  \left( \frac{z}{2} + i\delta \right) \leq -2\delta(i-1) + \left( \frac{z}{2} + i\delta \right)  \leq z$, where we also used that $\delta \leq 1$, which follows trivially from the other assumptions in the statement of the lemma.

Therefore, if for all $i$ (infinitely many), we have that $\mathcal{E}_i$ holds (i.e., $\cap_i \mathcal{E}_i$), then the random walk will have never hit $z$ and must have hit $y$ at some point meaning $X_T=y$, hence
\begin{equation}\label{eq:starfish}
\Pr(X_T=z) \leq  \Pr( \lor_i \lnot \mathcal{E}_i) \leq
\sum_i \Pr(  \lnot \mathcal{E}_i).
\end{equation}

Fix some $i^{\text{th}}$ interval---i.e., $t\in [\tau_{i-1}, \tau_{i})=[\sum_j^{i-1} t_j,\sum_j^{i} t_{j})$---and let $\tau=t_{i}$ be the interval length. 
Let $Y=X_{\tau_i} - X_{\tau_{i-1}}$.
We have $\E{Y} {\leq} -\delta \tau$. Thus, by Azuma-Hoeffding inequality, 
\begin{align} \label{eq:notEi}
\Pr(\lnot \mathcal{E}_i  ) =  \Pr\left(Y > \frac{-\delta \tau}{2}\right) = \Pr\left( Y > \E{Y} + \frac{\delta \tau}{2}\right) \leq \exp\left( - \frac{2 \delta^2 \tau^2}{2^2 \tau}\right) =\exp\left( -  \frac{\delta^2 t_{i}}{2}\right)
\end{align}

Therefore, choosing $k$ large enough such that the random walk must have hit $y$,  we have from \eqref{eq:starfish} and \eqref{eq:notEi}

\begin{align*}
\Pr(X_T = z) 
&\leq \Pr(\lnot \mathcal{E}_k \lor  \dots  \lor \lnot \mathcal{E}_0) \leq
\sum_i^k  \Pr(\lnot \mathcal{E}_i)\\
&\leq \sum_{i=0}^k \exp\left(- \frac{\delta^2 t_{i}}{2}\right) 
= \exp\left( - \frac{\delta^2 z}{4} \right)\sum_{i=0}^k 
\exp\left( -  \delta^2 i \frac{\delta}{4} \right)\\
&=
\exp\left( - \frac{\delta^2 z}{4} \right) \sum_{i=0}^k \exp\left( -    \frac{\delta^3}{4} \right)^i
\leq
\exp\left( - \frac{\delta^2 z}{4}\right)
\frac{1}{1-e^{-\delta^3/4}},
\end{align*}
where in the last line we used the geometric series.

\end{proof}

The following lemma is also used in the proof of Lemma~\ref{lem:32}.

\begin{lemma}[\cite{LENGLER_STEGER_2018}, Theorem 2.3]\label{thm:drift}
Let \( (Z_t)_{t \in \mathbb{N}_0} \) be a Markov chain with state space \( \mathcal{Z} \) and with a trace function \( \alpha : \mathcal{Z} \to \mathcal{S} \subseteq \mathbb{N}_0 \), and assume \( \alpha(Z_0) = n \). Let \( C \in \mathbb{N}_0 \) be some positive constant and let 
\[
T_C := \inf \{ t \in \mathbb{N}_0 : \alpha(Z_t) \leq C \} \;.
\]
Assume furthermore that there exists a constant \( c > 0 \) and a function \( g : \mathcal{S} \to \mathbb{R} \) such that \( g(x) = 0 \) for all \( x \leq C \), and \( g(x) > 0 \) for all \( x > C \), and such that for all \( t \geq 0 \)
\[
\mathbb{E}\big[ g(\alpha(Z_{t+1})) \mid Z_t = z \big] \leq g(\alpha(z)) - c \quad \text{for all } z \in \mathcal{Z} \text{ with } \alpha(z) > C \;.
\]
Then \( \mathbb{E}[T_C] \leq g(n)/c \).
\end{lemma}

\subsection{Proof of Corollary~\ref{cor:multi}}
\label{appA3}

\begin{proof}
Having defined $\mathbb{E}[T_i] \leq t^+$ and by Markov's inequality, for each agent \( i \), we have:
\[
\mathbb{P}(T_i \geq 2t^+) \leq \frac{\mathbb{E}[T_i]}{2t^+} \leq \frac{1}{2}.
\]
If an agent fails to deliver within \( 2t^+ \) timesteps, we restart the process. 
Therefore, each attempt is independent and has a success probability of at least \( \frac{1}{2} \).
The probability that an agent fails to deliver in \( f \) consecutive attempts, then, is at most:
\[
\left(\frac{1}{2}\right)^f.
\]

Define now indicator variables representing whether the agent failed: \( X_i := \mathbb{1}\{T_i > 2ft^+\} \), so that \( X_i \in \{0,1\} \), and \( \mathbb{E}[X_i] \leq (1/2)^f \).

Let \( S_n := \sum_{i=1}^n X_i \) be the number of agents that failed by time \( 2ft^+ \). We want to bound the probability that more than $25\%$ of $n$ agents indeed fail:
\[
\mathbb{P}\left( S_n \geq 0.25n \right).
\]

For this we can use the Chernoff bound:
\[
\Pr[S_n \geq R] \leq 2^{-R} \quad \text{for } R \geq 2e\, \mathbb{E}[S_n].
\]

Observe that, since we established \( \mathbb{E}[X_i] \leq (1/2)^f \), and for $f = 5$, we have

\[
\mathbb{E}[S_n] \leq \sum_i^n\mathbb{E}\left[X_i\right] \leq \sum_i^n\left(\frac{1}{2}\right)^f \leq \frac{n}{32}.
\]

Therefore, setting \( R = \frac{n}{4} \), we have:
\[
R = \frac{n}{4} \geq 8 \cdot \frac{n}{32} \geq 8 \mathbb{E}[S_n] \geq 2e\, \mathbb{E}[S_n].
\]

This yields:
\[
\mathbb{P}\left(S_n \geq \frac{n}{4}\right) \leq 2^{-n/4}.
\]

Hence, with probability at least \( 1 - 2^{-n/4} \), at least 75\% of the agents finish by time \( 10t^+ \).

\end{proof}

\end{document}